\PassOptionsToPackage{hypertexnames=false}{hyperref}
\documentclass{theoretics}

  \allowdisplaybreaks

\usepackage{xfrac}

\DeclareMathOperator*{\argmax}{arg\,max}

\newcommand{\poly}{\mathop{\mathrm{poly}}}

\newcommand{\F}{\mathbb{F}}

\newcommand{\E}{\mathbb{E}}

\newcommand{\eps}{\varepsilon}
\newcommand{\T}{\mathcal{T}}
\newcommand{\X}{\mathcal{X}}
\newcommand{\Y}{\mathcal{Y}}
\newcommand{\Z}{\mathcal{Z}}

\usepackage[capitalise,noabbrev,nameinlink]{cleveref}

\title{A Refined Laser Method and Faster Matrix Multiplication}

\ThCSshortnames{J.~Alman, V.~Vassilevska Williams} 
 \ThCSshorttitle{A Refined Laser Method and Faster Matrix Multiplication}

\ThCSauthor[1]{Josh Alman}{josh@cs.columbia.edu}[https://orcid.org/0009-0002-2204-1359]

\ThCSauthor[2]{Virginia Vassilevska Williams}{virgi@mit.edu}[https://orcid.org/0000-0003-4844-2863]

\ThCSaffil[1]{Department of Computer Science, Columbia University}

\ThCSaffil[2]{CSAIL and EECS, Massachusetts Institute of Technology}

\ThCSyear{2024}
\ThCSarticlenum{21}
\ThCSreceived{May 1, 2023}
\ThCSrevised{Apr 16, 2024}
\ThCSaccepted{May 11, 2024}
\ThCSpublished{Sep 4, 2024}
\ThCSdoicreatedtrue

\ThCSthanks{An abstract of this paper appeared at SODA 2021 \cite{AlmanW21}. Josh Alman was supported by a Michael O. Rabin Postdoctoral Fellowship. Virginia Vassilevska Williams was partially supported by an NSF Career Award, a Sloan Fellowship,
NSF Grants CCF-1909429 and CCF-1514339, and BSF Grant BSF:2012338. }

\ThCSkeywords{matrix multiplication, algorithm, laser method, tensor, omega}

\begin{document}

\maketitle

\begin{abstract}
The complexity of matrix multiplication is measured in terms of $\omega$, the smallest real number such that two $n\times n$ matrices can be multiplied using $O(n^{\omega+\epsilon})$ field operations for all $\epsilon>0$; the best bound until now is $\omega<2.37287$ [Le Gall'14]. All bounds on $\omega$ since 1986 have been obtained using the so-called laser method, a way to lower-bound the `value' of a tensor in designing matrix multiplication algorithms. The main result of this paper is a refinement of the laser method that improves the resulting value bound for most sufficiently large tensors. Thus, even before computing any specific values, it is clear that we achieve an improved bound on $\omega$, and we indeed obtain the best bound on $\omega$ to date:
$$\omega < 2.37286.$$
The improvement is of the same magnitude as the improvement that [Le Gall'14] obtained over the previous bound [Vassilevska W.'12]. Our improvement to the laser method is quite general, and we believe it will have further applications in arithmetic complexity.
\end{abstract}

\section{Introduction}
Settling the algorithmic complexity of matrix multiplication is one of the most fascinating open problems in theoretical computer science. The main measure of progress on the problem is the exponent $\omega$, defined as the smallest real number for which $n\times n$ matrices over a field can be multiplied using $O(n^{\omega+\eps})$ field operations, for every $\eps>0$. The value of $\omega$ could depend on the field, although the algorithms we discuss in this paper work over any field.
The straightforward algorithm shows that $\omega\leq 3$, and we can see that $\omega\geq 2$ since any algorithm must output $n^2$ entries. In 1969, Strassen \cite{strassen} obtained the first nontrivial upper bound on $\omega$, showing that $\omega<2.81$. Since then, a long series of papers (e.g. \cite{Pan78,BCRL79,Pan80,Sch81,Romani82,cw81as,laser,coppersmith,stothers,v12,legall,cohn2003group,cohn2005group,cohn2013fast}) has developed a powerful toolbox of techniques, culminating in the best bound to date of $\omega<2.37287$.

In this paper, we add one more tool to the toolbox and lower the best bound on the matrix multiplication exponent to $$\omega<2.37286.$$ 

The main contribution of this paper is a new refined version of the {\em laser method} which we then use to obtain the new bound on $\omega$. The laser method (as coined by Strassen \cite{laser}) is a powerful mathematical technique for analyzing tensors. In our context, it is used to lower bound the ``value'' of a tensor in designing matrix multiplication algorithms. The laser method also has applications beyond bounding $\omega$ itself, including to other problems in arithmetic complexity like computing the ``asymptotic subrank'' of tensors~\cite{Alman19}, and to problems in extremal combinatorics like constructing tri-colored sum-free sets~\cite{kleinberg}. We believe our improved laser method may have other diverse applications.

We will see that our new method achieves better results than the laser method of prior work when applied to almost any sufficiently large tensor, including most of the tensors which arise in prior bounds on $\omega$. In fact, unlike in other recent work, it is clear before running any code that our new method yields a new improved bound on $\omega$.

The last several improvements to the best bound on $\omega$ have been modest. Most recently, Le Gall~\cite{legall} brought the upper bound on $\omega$ from $2.37288$~\cite{v12} to $2.37287$, and we bring it down to $2.37286$. (More precisely, we prove $\omega < 2.3728596$.)
A recent line of work has shown that only modest improvements can be expected if one continues using similar techniques. All fast matrix multiplication algorithms since 1986 use the laser method applied to the Coppersmith-Winograd family of tensors \cite{coppersmith}. The techniques can also be simulated within the group theoretic method of Cohn and Umans \cite{cohn2003group,cohn2005group,cohn2013fast}. 
It is known that the laser method, even our refined version of it, when applied to powers of the particular Coppersmith-Winograd tensor $CW^{\otimes 32}_5$ that achieves the current best bounds on $\omega$, cannot achieve $\omega<2.3725$ \cite{ambainis}. 
Furthermore, a sequence of several papers \cite{ambainis,sunfl,blasiak,blasiak2017groups,almanitcs,aw2,Alman19,ChristandlVZ19,blaser2020slice,wojtala2023irreversibility,blasiak2023matrix} has given strong limitations on the power of the laser method, the group theoretic method and their generalizations, showing that for many natural families of tensors, even approaches substantially more general than the laser method would not be able to prove that $\omega=2$.

That said, the known limitation results are very specific to the tensors they are applied to, and so it is not ruled out that our improved laser method could be applied to a different family of tensors to yield even further improved bounds on $\omega$.
Even in the case of the smaller Coppersmith-Winograd tensor $CW_1$, the known limitations are weaker, and it is not ruled out that one could achieve $\omega<2.239$ using it (see \cite[Table 1]{ambainis}).

\subsection{The Laser Method and Our Improvement}

In this subsection, we give an overview of our improvement to the laser method. We assume familiarity with basic notions related to tensors and matrix multiplication; unfamiliar readers may want to read Section~\ref{sec:prelims} first.

Fast matrix multiplication algorithms since the 1980s have been designed by making use of a cleverly-chosen intermediate tensor $T$. They have consisted of two main ingredients:
\begin{enumerate}
    \item An algorithm for efficiently computing $T$ (i.e., a proof that $T$ has low asymptotic rank $\tilde{R}(T)$), and
    \item A proof that $T$ has a high \emph{value} for computing matrix multiplication (i.e., a restriction of $T^{\otimes n}$ into a large direct sum of matrix multiplication tensors).
\end{enumerate}
Since the work of Coppersmith and Winograd~\cite{coppersmith}, the fastest matrix multiplication algorithms have used $T=CW_q$, the \emph{Coppersmith-Winograd tensor}. Coppersmith and Winograd showed that the asymptotic rank of $CW_q$ is as low as possible given its dimensions. Hence, subsequent work has focused on improving the bound on the value of $CW_q$, and this is the approach we take as well.

The primary way that past work has bounded the values of tensors like $CW_q$ is using the \emph{laser method}. The laser method was so-named by Strassen~\cite{strassenlaser1}, and then further developed by Coppersmith and Winograd~\cite{coppersmith} into its current form. We now describe the laser method at a high level.

Consider a tensor $T$ over finite variables sets $X,Y,Z$, given by $$T = \sum_{x \in X} \sum_{y \in Y} \sum_{z \in Z} a_{xyz} \cdot x y z$$ for some coefficients $a_{xyz} \in \F$ from the underlying field $\F$. Let us partition the variable sets as $X = X_1 \cup \cdots \cup X_{k_X}$, $Y = Y_1 \cup \cdots \cup Y_{k_Y}$, and $Z = Z_1 \cup \cdots \cup Z_{k_Z}$. Hence, letting $$T_{ijk} = \sum_{x \in X_i} \sum_{y \in Y_j} \sum_{z \in Z_k} a_{xyz} \cdot x y z$$ be the subtensor of $T$ restricted to $X_i,Y_j,Z_k$, we have $$T = \sum_{i=1}^{k_X} \sum_{j=1}^{k_Y} \sum_{k=1}^{k_Z} T_{ijk}.$$

Suppose for now that each $T_{ijk}$ is either $0$ or else a matrix multiplication tensor; this will simplify the presentation here but is not needed in the general setting. Hence, $T$ is a sum of matrix multiplication tensors. The typical way to obtain matrix multiplication algorithms from a sum of matrix multiplication tensors, however, requires the sum to be a {\em direct sum}, and $T$ is not a direct sum in general.
One could zero-out many of the $X_i$, $Y_j$, and $Z_k$ parts until $T$ is a direct sum of the remaining $T_{ijk}$ subtensors, but this typically removes `too much' from the tensor. 

The laser method instead takes the following approach. First, we pick a probability distribution $\alpha$ on the nonzero subtensors of $T$, which assigns probability $\alpha_{ijk}$ to $T_{ijk}$. Next, we take a large Kronecker power $n$ of $T$: $$T^{\otimes n} = \sum_{(T_1, T_2, \ldots, T_n) \in \{T_{ijk} \mid i \in [k_X], j \in [k_Y], k \in [k_Z]\}^n} T_1 \otimes T_2 \otimes \cdots \otimes T_n.$$ The goal of the laser method is to zero-out variables in $T^{\otimes n}$ so that what remains is a direct sum of $B$ different subtensors of the form $T_1 \otimes T_2 \otimes \cdots \otimes T_n$ which are \emph{consistent with $\alpha$}, meaning, for each $i \in [k_X], j \in [k_Y], k \in [k_Z]$, we have $|\{ \ell \in [n] \mid T_\ell = T_{ijk}\}| = \alpha_{ijk} \cdot n$. Each subtensor $T_1 \otimes T_2 \otimes \cdots \otimes T_n$ is itself a matrix multiplication tensor, so this is a desired direct sum. 

How large can we hope for $B$ to be? One upper bound is in terms of the \emph{marginals} of $\alpha$. For each $i \in [k_X]$, write $\alpha_{X_i} := \sum_{j \in [k_Y], k \in [k_Z]} \alpha_{ijk}$. Each $T_1 \otimes T_2 \otimes \cdots \otimes T_n$ which is consistent with $\alpha$ uses the variables from a set $X_{a_1} \times X_{a_2} \times \cdots \times X_{a_n}$ where, for each $i \in [k_X]$, we have $|\{ \ell \in [n] \mid a_\ell = i\}| = \alpha_{X_i} \cdot n$. The number of such sets is given by the multinomial coefficient $$\binom{n}{\alpha_{X_1}n, \alpha_{X_2}n, \ldots, \alpha_{X_{k_X}}n}.$$
Since each of these sets can be used by at most one of the final $B$ subtensors, this multinomial coefficient upper bounds $B$. The laser method shows that if $T$ and $\alpha$ satisfy some additional conditions, then roughly this bound on $B$ can actually be achieved!

One of these conditions, which is the focus of our new improvement, is on the marginals of $\alpha$. Let $D_\alpha$ be the set of probability distributions $\beta$ with the same marginals as $\alpha$ (i.e., with $\alpha_{X_i} = \beta_{X_i}$ for all $i \in [k_X]$, $\alpha_{Y_j} = \beta_{Y_j}$ for all $j \in [k_Y]$, and $\alpha_{Z_k} = \beta_{Z_k}$ for all $k \in [k_Z]$). The laser method only achieves the aforementioned value of $B$ if there are no other probability distributions $\beta$ in $D_\alpha$. This is intuitively because the laser method is only zeroing-out sets of variables, and so it cannot distinguish between two distributions which have the same marginals. This is not an issue when analyzing smaller powers of $CW_q$ as in~\cite{coppersmith,stothers}, since there the tensors and partitions are small enough that the linear system defining $D_\alpha$ has full rank. However, it becomes an issue when analyzing larger powers of $CW_q$ as in~\cite{v12,legall}.

The prior work dealt with this issue in a greedy way: Suppose that after the zeroing outs described above, we are left with a direct sum of $B$ subtensors consistent with $\alpha$, plus $m \cdot B$ other subtensors which are consistent with other distributions in $D_\alpha$. We can repeatedly pick a subtensor $S$ consistent with $\alpha$, and zero-out roughly $m$ subtensors which aren't consistent with $\alpha$ until $S$ no longer shares variables with any remaining subtensors. We can then keep $S$ as an independent subtensor, but we may have zeroed out roughly $m$ other subtensors consistent with $\alpha$ in the process. We repeat until only subtensors consistent with $\alpha$ remain. Hence, the old approach leaves us with a direct sum of $$\Theta\left(\frac{B}{m}\right)$$ subtensors consistent with $\alpha$.

In this work, we present a new way to deal with this issue, which improves the number of subtensors consistent with $\alpha$ at the end of the laser method to $$\Theta\left( \frac{B}{\sqrt{m}}\right).$$ 
This directly improves the final value bound achieved by the laser method by a factor of $\Theta(\sqrt{m})$. Since most of the applications of the laser method in the previous best bounds on $\omega$~\cite{v12,legall} apply it in settings where $m \gg 1$ is large enough to reduce the final value, it is evident \emph{even before running any code or computing any specific values} that our improvement on the laser method leads to an improved bound on $\omega$.

Similar to other steps of the laser method, our new construction is probabilistic. We show that if one picks a \emph{random} subset of roughly $B / \sqrt{m}$ subtensors consistent with $\alpha$, and zeroes out all variables which aren't used by any of them, then there is a nonzero probability that all other subtensors are zeroed out.

It is worth asking whether the factor of $\sqrt{m}$ can be further improved. We give evidence that an improvement is not possible by constructing tensors for which our new probabilistic argument is tight up to low-order factors.  Our constructed tensors even share special properties with the tensors that the laser method is usually applied to (they are ``free'').
That said, we leave open the possibility of improving the $\sqrt m$ bound for the specific tensors to which the laser method ultimately applies our probabilistic argument. 

We present our new probabilistic argument for dealing with distributions $\beta \in D_\alpha$ other than $\alpha$ in Section~\ref{sec:newidea}, and then we show how to incorporate it into the laser method in Section~\ref{sec:proof}. We then get into the details of actually applying the laser method to $CW_q$ to achieve our new bound on $\omega$: In Section~\ref{sec:heuristics} we discuss the computational problem of applying the laser method to a given tensor, and some algorithms and heuristics for solving it, and in Section~\ref{sec:final-answer} we detail how to apply the laser method to $CW_q$ specifically. Our new bound of $\omega<2.3728596$ is achieved by applying our refined laser method to $CW_5^{\otimes 32}$, the same tensor used by~\cite{legall}.

Our primary new contribution is the improved factor of $\sqrt{m}$ in the laser method, but we do add some new heuristics to the optimization framework of~\cite{v12,legall} for applying the laser method to $CW_q$ as well. Although many of the ideas in our proof are similar to past work, we nonetheless give all the details, and we have written the body of the paper assuming little prior knowledge from the reader. We recommend the reader go through the sections in order, as the notation needed to apply the laser method is built up throughout.

\subsection{Other Related Work}

\paragraph*{Rectangular Matrix Multiplication}

We do not specifically address algorithms for rectangular matrix multiplication in this paper. Since the best known algorithms for rectangular matrix multiplication also make use of the laser method, our techniques can be used to design faster rectangular matrix multiplication algorithms as well. That said, the computational issues which arise when bounding the running time of rectangular matrix multiplication are even more severe than when bounding $\omega$, and in fact the best known algorithms only use the $4$th Kronecker power of $CW_q$~\cite{legallrect2}, where our improved laser method does not yet kick in. 

\paragraph*{Lower Bounds for Matrix Multiplication}

As mentioned earlier, there are a number of different limitation results showing that certain techniques cannot be used to prove $\omega=2$. However, there is little known in the way of \emph{unconditional} lower bounds against matrix multiplication. Raz~\cite{raz2002complexity} showed that the restricted class of arithmetic circuits with \emph{bounded coefficients} for $n \times n \times n$ matrix multiplication over $\mathbb{C}$ requires size $\Omega(n^2 \log n)$. Another line of work \cite{strassen1983rank,lickteig1984note,burgisser2011geometric,landsberg2015new,landsberg20172} has shown border rank lower bounds for matrix multiplication tensors. The current best bound by Landsberg and Micha{\l}ek~\cite{landsberg20172} shows that the $n \times n \times n$ matrix multiplication tensor over $\mathbb{C}$ has border rank at least $2n^2 - \log_2(n) - 1$. Relatedly, the best known bound in a line of work on rank lower bounds~\cite{blaser19995,shpilka2003lower,landsberg2014new} shows that the $n \times n \times n$ matrix multiplication tensor over $\mathbb{C}$ has rank at least $3n^2 - o(n^2)$~\cite{landsberg2014new}.

\paragraph*{Subsequent Algorithms for Matrix Multiplication}

After the preliminary version of this paper, subsequent work~\cite{duan2023faster,williams2024new} designed a further improved matrix multiplication algorithm, achieving $\omega < 2.371552$. The key idea behind these improvements is a new \emph{asymmetric} way to apply the laser method to $CW_q$ and its subtensors. Interestingly, their analysis, and particularly a new observation called ``combination loss''~\cite{duan2023faster}, appears quite specific to the tensor $CW_q$, and it is unclear if the same approach would yield improvements for other tensors to which the laser method applies. It also appears difficult to use this asymmetric approach in conjunction with our refined laser method. Roughly, the new zeroing outs of the refined laser method may interfere with other subtensors which are intended to be kept in the asymmetric approach; we refer the reader to \cite[End of Section 2.2]{williams2024new} for a more technical discussion.

\section{Preliminaries} \label{sec:prelims}

\subsection{Notation}

For a positive integer $n$, we write $[n] := \{1,2,3,\ldots,n\}$. For a set $S$ and nonnegative integer $k$, we write $\binom{S}{k} := \{ T \subseteq S \mid |T|=k\}$. For a set $S$, positive integer $d$, index $i \in [d]$, and vector $x \in S^d$, we write $x_i$ for entry $i$ of $x$. Any other uses of subscripts should be clear from context.

For nonnegative integers $a_1, \ldots, a_k$ with $a_1 + \cdots + a_k = n$, we write $\binom{n}{[a_i]_{i \in [k]}} := \binom{n}{a_1,a_2,\ldots,a_k}$ for the multinomial coefficient. One standard bound on multinomial coefficients that we will frequently make use of is that, if $p_1, \ldots, p_k \in [0,1]$ sum to $p_1 + \cdots + p_k = 1$, then for sufficiently large positive integers $n$ such that $p_i \cdot n$ is an integer for all $i \in [k]$, we have that $$\binom{n}{[p_i \cdot n]_{i \in [k]}} = \left( \prod_{i=1}^k p_i^{-p_i} \right)^{n - o(n)}.$$

\subsection{Tensors and Sums}

Let $\F$ be any field, and $X = \{x_1, \ldots, x_{|X|}\},Y=\{y_1,\ldots,y_{|Y|}\},Z=\{z_1, \ldots, z_{|Z|}\}$ be finite sets (which we will refer to as sets of \emph{variables}). A tensor $T$ over $X,Y,Z$ is a trilinear form $$T = \sum_{i=1}^{|X|} \sum_{j=1}^{|Y|} \sum_{k=1}^{|Z|} a_{ijk} \cdot x_i  y_j  z_k,$$ where the $a_{ijk}$ are coefficients from the field $\F$. We'll focus in particular on tensors $T$ whose coefficients $a_{ijk}$ are all $0$ or $1$, so that $T$ can be thought of as a tensor over any field. Such tensors $T$ can be thought of as subsets of $X \times Y \times Z$. For $x_i \in X, y_j \in Y, z_k \in Z$, we say $x_i  y_j  z_k$ has nonzero coefficient in $T$ if $a_{ijk} \neq 0$.

For tensor $T$ over $X = \{x_1, \ldots, x_{|X|}\},Y=\{y_1,\ldots,y_{|Y|}\},Z=\{z_1, \ldots, z_{|Z|}\}$ and tensor $T'$ over $X' = \{x'_1, \ldots, x'_{|X'|}\},Y'=\{y'_1,\ldots,y'_{|Y'|}\},Z'=\{z'_1, \ldots, z'_{|Z'|}\}$,  given by $$T = \sum_{i=1}^{|X|} \sum_{j=1}^{|Y|} \sum_{k=1}^{|Z|} a_{ijk} \cdot x_i  y_j  z_k, \hspace{30pt} T' = \sum_{i'=1}^{|X'|} \sum_{j'=1}^{|Y'|} \sum_{k'=1}^{|Z'|} b_{i'j'k'} \cdot x'_{i'}  y'_{j'}  z'_{k'},$$ we now describe a number of operations and relations. We will use these two tensors as running notation throughout this section.

If $X=X'$, $Y=Y'$, and $Z=Z'$, then the \emph{sum} $T+T'$ is the tensor whose coefficient of $x_iy_jz_k$ is $a_{ijk} + b_{ijk}$. With our view of tensors as polynomials, this is the usual way to sum $T$ and $T'$.

The \emph{direct sum} $T \oplus T'$ is the sum $T+T'$ over the disjoint unions $X \sqcup X', Y \sqcup Y', Z \sqcup Z'$. In other words, it is the sum $T+T'$ after we first relabel the sets of variables so that they are disjoint.

We say $T$ and $T'$ are \emph{isomorphic}, written $T \equiv T'$, if one can get from one to the other by renaming variables. In other words, they are isomorphic if $|X|=|X'|$, $|Y|=|Y'|$, $|Z|=|Z'|$, and there are bijections $\pi_X : [|X|] \to [|X'|]$, $\pi_Y : [|Y|] \to [|Y'|]$, and $\pi_Z : [|Z|] \to [|Z'|]$ such that $a_{ijk} = b_{\pi_x(i) \pi_y(j) \pi_z(k)}$ for all $i \in [|X|]$, $j \in [|Y|]$, and $k \in [|Z|]$.

The \emph{rotation} of $T$, denoted $T^r$, is the tensor over $Y,Z,X$ such that the coefficient of $y_jz_kx_i$ in $T^r$ is $a_{ijk}$. We similarly write $T^{rr}$ for the corresponding tensor over $Z,X,Y$.

\subsection{Kronecker Products}

The \emph{Kronecker product} $T \otimes T'$ is a tensor over $X \times X', Y \times Y', Z \times Z'$ given by
$$T \otimes T' = \sum_{i=1}^{|X|} \sum_{j=1}^{|Y|} \sum_{k=1}^{|Z|} \sum_{i'=1}^{|X'|} \sum_{j'=1}^{|Y'|} \sum_{k'=1}^{|Z'|} a_{ijk} \cdot b_{i'j'k'} \cdot (x_i,x'_{i'}) \cdot (y_j,y'_{j'}) \cdot (z_k,z'_{k'}).$$ One can think of $T \otimes T'$ as multiplying $T$ and $T'$ as polynomials, but then `merging' together pairs of $x$-variables, pairs of $y$-variables, and pairs of $z$-variables into single variables.

For a positive integer $n$, we write $T^{\otimes n} := T \otimes T \otimes T \otimes \cdots \otimes T$ ($n$ times) for the \emph{Kronecker power} of $T$, which is a tensor over $X^n, Y^n, Z^n$. For $I \in [|X|]^n$, we will write $x_I$ to denote the element $(x_{I_1}, x_{I_2}, \ldots, x_{I_n}) \in X^n$, and similarly for $Y^n,Z^n$, so that
$$T^{\otimes n} = \sum_{I \in [|X|]^n} \sum_{J \in [|Y|]^n} \sum_{K \in [|Z|]^n} \left( \prod_{\ell=1}^n a_{I_\ell J_\ell K_\ell}  \right) \cdot x_I  y_J  z_K.$$

\subsection{Tensor Rank}

We say $T$ has rank $1$ if we can write $$T = \left( \sum_{i=1}^{|X|} \alpha_i \cdot x_i \right) \cdot \left( \sum_{j=1}^{|Y|} \beta_j \cdot y_j \right) \cdot \left( \sum_{k=1}^{|Z|} \gamma_k \cdot z_k \right)$$ for some values $\alpha_i, \beta_j, \gamma_k \in \F$. Equivalently, $T$ has rank $1$ if there exist $\alpha_i, \beta_j, \gamma_k \in \F$  such that $a_{ijk} = \alpha_i \cdot \beta_j \cdot \gamma_k$ for all $i,j,k$. The rank of a tensor $T$, denoted $R(T)$, is the minimum nonnegative integer such that there are rank $1$ tensors $T_1, \ldots, T_{R(T)}$ over $X,Y,Z$ with $T_1 + \cdots + T_{R(T)} = T$. This is analogous to the rank of a matrix. We call the sum $T_1 + \cdots + T_{R(T)}$ a \emph{rank $R(T)$ expression} for $T$. For tensors $T,T'$, rank satisfies the basic properties:
\begin{itemize}
    \item $R(T+T') \leq R(T) + R(T')$, by adding the rank expressions for $T$ and $T'$,
    \item $R(T)=R(T^r)$ by rotating the rank expression, and
    \item $R(T \otimes T') \leq R(T) \cdot R(T')$, by the distributive property, since one can verify that the Kronecker product of two rank $1$ tensors is also a rank $1$ tensor. 
\end{itemize}
This third property inspires the definition of the \emph{asymptotic rank} $\tilde{R}(T)$ of $T$, given by $$\tilde{R}(T) = \lim_{n \to \infty} (R(T^{\otimes n}))^{1/n}.$$ By Fekete's lemma, $\tilde{R}(T)$ is well-defined, and is upper-bounded by $(R(T^{\otimes m}))^{1/m}$ for any fixed positive integer $m$. As we will see, there are many tensors $T$ for which $R(T) > \tilde{R}(T)$, and this is crucial in the study of matrix multiplication algorithms.

\subsection{Matrix Multiplication Tensors}

For positive integers $a,b,c$, the $a \times b \times c$ matrix multiplication tensor, written $\langle a,b,c \rangle$, is a tensor over $\{x_{ij}\}_{i \in [a], j \in [b]}$, $\{y_{jk}\}_{j \in [b], k \in [c]}$, $\{z_{ki}\}_{k \in [c], i \in [a]}$, given by $$\langle a,b,c \rangle = \sum_{i \in [a]} \sum_{j \in [b]} \sum_{k \in [c]} x_{ij} y_{jk} z_{ki}.$$
Note that $\langle a,b,c \rangle^r \equiv \langle b,c,a \rangle$. The tensor $\langle a,b,c \rangle$ is the trilinear form which one evaluates when multiplying an $a \times b$ matrix with a $b \times c$ matrix. In other words, for $A \in \F^{a \times b}$ and $B \in \F^{b \times c}$, if we substitute the $(i,j)$ entry of $A$ for $x_{ij}$ and the $(j,k)$ entry of $B$ for $y_{jk}$, then the resulting coefficient of $z_{ki}$ in $\langle a,b,c \rangle$ is the $(i,k)$ entry of the matrix product $A \times B$.

One can verify that for any positive integers $a,b,c,d,e,f$ we have $\langle a,b,c \rangle \otimes \langle d,e,f \rangle \equiv \langle ad, be, cf \rangle$. This corresponds to the fact that block matrices can be multiplied by appropriately multiplying and adding together blocks.

For any positive integers $q,r$, if $R(\langle q,q,q \rangle) = r$, then one can use the corresponding rank expression to design an arithmetic circuit for $n \times n \times n$ matrix multiplication of size $O(n^{\log_q(r)})$. This follows from the recursive approach introduced by Strassen~\cite{strassen}; see e.g. \cite[Proposition~1.1, Theorem~5.2]{blaser}. The exponent of matrix multiplication, $\omega$, is hence defined as $$\omega := \inf_{q \in \mathbb{N}} \log_q R(\langle q,q,q \rangle).$$ Thus, using Strassen's recursive approach, for every $\eps>0$, there is an arithmetic circuit for $n \times n \times n$ matrix multiplication of size $O(n^{\omega + \eps})$. For instance, Strassen showed that $R(\langle 2,2,2 \rangle) \leq 7$, which implied $\omega \leq \log_2(7) < 2.81$. Since $\langle q,q,q \rangle^{\otimes n} \equiv \langle q^n,q^n,q^n \rangle$, we can equivalently write that, for any fixed integer $q \geq 2$, $$\omega = \log_q \tilde{R}(\langle q,q,q \rangle).$$
The value of $\omega$ could depend\footnote{In fact, it is known that $\omega$ depends only on the characteristic of $\F$~\cite{Sch81}.} on the field $\F$, although the bounds we give in this paper work over any field. Note that bounding the rank of a rectangular matrix multiplication tensor can also yield bounds on $\omega$: if $R(\langle a,b,c \rangle) \leq r$, then by symmetry, $R(\langle b,c,a \rangle) \leq r$ and $R(\langle c,a,b \rangle) \leq r$, and so taking the Kronecker product of the three, we see $R(\langle abc,abc,abc \rangle) \leq r^3$ which yields $\omega \leq 3 \log_{abc} r$.

\subsection{Sch{\"o}nhage's Asymptotic Sum Inequality}

By definition, in order to upper bound $\omega$, it suffices to upper bound the (asymptotic) rank of some matrix multiplication tensor. Sch{\"o}nhage~\cite{Sch81} showed that it also suffices to upper bound the (asymptotic) rank of a \emph{direct sum} of multiple matrix multiplication tensors.

\begin{theorem}[\cite{Sch81}]\label{thm:schonhage}
Suppose there are positive integers $r>m$, and $a_i,b_i,c_i$ for $i \in [m]$, such that the tensor $$T = \bigoplus_{i=1}^m \langle a_i,b_i,c_i \rangle$$ has $\tilde{R}(T) \leq r$. Then, $\omega \leq 3 \tau$, where $\tau \in [2/3,1]$ is the solution to $$\sum_{i=1}^m (a_i \cdot b_i \cdot c_i)^\tau = r.$$
\end{theorem}

\subsection{Zeroing-Outs, Restrictions, and Value}

We say $T'$ is a \emph{restriction} of $T$ if there is an $\F$-linear map $A : \textsf{span}_\F(X) \to \textsf{span}_\F(X')$, which maps $\F$-linear combinations of variables of $X$ to $\F$-linear combinations of variables of $X'$, and similarly $\F$-linear maps $B : \textsf{span}_\F(Y) \to \textsf{span}_\F(Y')$, and $C : \textsf{span}_\F(Z) \to \textsf{span}_\F(Z')$, such that  $$T' = \sum_{i=1}^{|X|} \sum_{j=1}^{|Y|} \sum_{k=1}^{|Z|} a_{ijk} \cdot A(x_i) \cdot B(y_j) \cdot C(z_k).$$
It is not hard to verify that if $T'$ is a restriction of $T$, then $R(T) \geq R(T')$, and $\tilde{R}(T) \geq \tilde{R}(T')$, since the restriction of a rank $1$ tensor is still a rank $1$ tensor. Recent progress on bounding $\omega$ has worked by cleverly picking a tensor $T$, showing that $\tilde{R}(T)$ is `small', and showing that a `large' direct sum of matrix multiplication tensors is a restriction of a power $T^{\otimes n}$. We will follow this approach as well.

In fact, we will only use a limited type of restriction called a \emph{zeroing out}. We say $T'$ is a zeroing out of $T$ if $X' \subseteq X$, $Y' \subseteq Y$, $Z' \subseteq Z$, and the coefficient of $x_i y_j z_k$ is the same in $T$ and $T'$ for every $x_i \in X'$, $y_j \in Y'$, and $z_k \in Z'$. In this case, we write $T' = T|_{X',Y',Z'}$. We say the variables in $X \setminus X'$, $Y \setminus Y'$, and $Z \setminus Z'$ have been \emph{zeroed-out}; one can think of substituting in $0$ for those variables in $T$ to get to $T'$.

Coppersmith and Winograd~\cite{coppersmith} formalized this approach to bounding $\omega$ by defining the \emph{value} of a tensor. For $\tau \in [2/3,1]$, the $\tau$-value of $T$, written $V_\tau(T)$, is given by the supremum over all positive integers $n$, and all tensors of the form $\bigoplus_{i=1}^m \langle a_i,b_i,c_i \rangle$ which are restrictions of $(T \otimes T^r \otimes T^{rr})^{\otimes n}$, of $$\left( \sum_{i=1}^m (a_i \cdot b_i \cdot c_i)^\tau \right)^{\frac{1}{3n}}.$$ When $\tau$ is clear from context, we will simply write $V(T)$ and call it the value of $T$. One can see that for tensors $T,T'$, the value $V_\tau$ satisfies $V_\tau(T\otimes T')\geq V_\tau(T) \cdot V_\tau(T')$ and $V_\tau(T\oplus T')\geq V_\tau(T)+V_\tau(T')$. We can also see that $V_\tau(\langle a,b,c \rangle) = (abc)^\tau$. We work with the tensor $T \otimes T^r \otimes T^{rr}$ instead of just $T$ in the definition of $V_\tau(T)$ since this more symmetric form can sometimes substantially increase the value of relatively `asymmetric' tensors.

By Theorem~\ref{thm:schonhage}, we get almost immediately that for any tensor $T$, and any $\tau \in [2/3,1]$, if $V_\tau(T) \geq \tilde{R}(T)$, then $\omega \leq 3 \tau$. Thus, in this paper, we focus on lower-bounding the values of certain tensors. We will give a recursive approach where the value of a tensor with certain structure can be bounded in terms of the values of its subtensors.

\subsection{Coppersmith-Winograd Tensors}

For a nonnegative integer $q$, the Coppersmith-Winograd tensor $CW_q$ is a tensor over $\{x_0,\ldots,x_{q+1}\}$, $\{y_0, \ldots, y_{q+1}\}$, $\{z_0, \ldots, z_{q+1}\}$ given by $$CW_q := x_0 y_0 z_{q+1} + x_0 z_{q+1} y_0 + x_{q+1} y_0 z_0 + \sum_{i=1}^q (x_0 y_i z_i + x_i y_0 z_i + x_i y_i z_0).$$ 
Notice in particular that $$\langle 1,1,q \rangle \equiv \sum_{i=1}^q x_0 y_i z_i, \hspace{20pt} \langle q,1,1 \rangle \equiv \sum_{i=1}^q x_i y_0 z_i, \hspace{20pt} \langle 1,q,1 \rangle \equiv \sum_{i=1}^q x_i y_i z_0,$$
so $CW_q$ is the sum of three matrix multiplication tensors and three `corner terms' (which are also $\langle 1,1,1 \rangle$ matrix multiplication tensors). 
Coppersmith and Winograd~\cite{coppersmith} showed\footnote{In fact, they showed that the \emph{`border rank'} of $CW_q$ is $\leq q+2$, and border rank is known to upper bound asymptotic rank~\cite{bini1980border}. The fact that $\tilde{R}(CW_q) \geq q+2$ follows since $CW_q$ is a `concise' tensor; see, e.g.,~\cite[Remark 14.38]{burgisser2013algebraic}.} that $\tilde{R}(CW_q) = q+2$. The upper bounds on $\omega$ since Coppersmith and Winograd's work~\cite{coppersmith,stothers,v12,legall} have all been proved by giving value lower bounds for $CW_q$. Our improvement in this paper will come from further improving these value bounds.

\subsection{Salem-Spencer Sets}

The final technical ingredient from past work that we will need is a construction of large subsets of $\mathbb{Z}_M$ which avoid three-term arithmetic progressions.

\begin{theorem}[\cite{salemspencer,behrend1946sets}] \label{thm:SS}
For every positive integer $M$, there is a subset $A \subseteq \mathbb{Z}_M$ of size $|A| \geq M \cdot e^{-O(\sqrt{\log M})}$ such that any $a,b,c \in A$ satisfy $a+b=2c \pmod{M}$ if and only if $a=b=c$.
\end{theorem}

We briefly mention a `tensor interpretation' of Theorem~\ref{thm:SS}. For odd prime $M$, define the tensor $C_M$ over $\{x_0, \ldots, x_{M-1}\}$, $\{y_0, \ldots, y_{M-1}\}$, $\{z_0, \ldots, z_{M-1}\}$ by $$C_M := \sum_{i=0}^{M-1} \sum_{j=0}^{M-1} x_i \cdot y_j \cdot z_{(i+j)/2 \pmod{M}}.$$ Letting $A$ be the set from Theorem~\ref{thm:SS}, if we zero-out all $x_i$, $y_i$, and $z_i$ for which $i \notin A$ in $C_M$, then the result is the tensor $$\sum_{i \in A} x_i \cdot y_i \cdot z_{i},$$ which is a direct sum of $|A|\geq M \cdot e^{-O(\sqrt{\log M})}$ terms.

\section{Diagonalizing Arbitrary Tensors with Zeroing Outs} \label{sec:newidea}

We now present a main new technical tool which we will later use in proving value lower bounds for tensors. It can be thought of as a generalization of Theorem~\ref{thm:SS} to tensors beyond just $C_M$. The resulting bound we get is not as large (it yields $\Omega(\sqrt{M})$ instead of $M \cdot e^{-O(\sqrt{\log M})}$ for $C_M$), although we show later in Theorem~\ref{thm:opt1} and Theorem~\ref{thm:opt2} that such a loss is required for this more general statement. The proof uses the probabilistic method.

\begin{theorem} \label{thm:newidea}
Suppose $T$ is a tensor over $X,Y,Z$, with partitions $X = X_1 \cup X_2 \cup \cdots \cup X_{n}$, $Y = Y_1 \cup Y_2 \cup \cdots \cup Y_{n}$, $Z = Z_1 \cup Z_2 \cup \cdots \cup Z_{n}$ for some positive integer $n$. For $i,j,k \in [n]$, write $T_{ijk} := T|_{X_i, Y_j, Z_k}$. Let $S = \{(i,j,k) \in [n]^3 \mid T_{ijk} \neq 0 \}$, and suppose that:
\begin{itemize}
    \item $(i,i,i) \in S$ for all $i \in [n]$, and
    \item for all other $(i,j,k) \in S$, the three values $i,j,k$ are distinct.
\end{itemize}
Write $m := (|S|-n)/n$, and suppose that $m \geq 1$. Then, there is a subset $I \subseteq [n]$ of size $|I| \geq \frac{2n}{3 \sqrt{3m}}$ such that $T$ has a zeroing out into $\sum_{i \in I} T_{iii}$.
\end{theorem}

\begin{proof}
We use the probabilistic method. Let $p := \frac{1}{\sqrt{3m}}$, and let $R$ be a random subset of $[n]$ where each element is included independently with probability $p$, so $\E[|R|] = p \cdot n$. Define $S' \subseteq S$ by $$S' = \{ (i,j,k) \in S \mid i\neq j \neq k \neq i \text{ and } i,j,k \in R\}.$$ Any given $(i,j,k) \in S$ with $i\neq j \neq k \neq i$ is included in $S'$ with probability $p^3$, and so $\E[|S'|] = p^3 \cdot (|S|-n) = p^3 \cdot n \cdot m$. Let $A := \{ i \mid (i,j,k) \in S'\}$, and let $I := R \setminus A$. We have that
$$\E[|I|] \geq \E[|R|] - \E[|A|] \geq \E[|R|] - \E[|S'|] = p \cdot n - p^3 \cdot n \cdot m = n\cdot (p - p^3 m) = \frac{2n}{3 \sqrt{3m}}.$$
It follows that there is a choice of randomness with $|I| \geq \frac{2n}{3 \sqrt{3m}}$. Fix this choice, then let $T'$ be $T$ after zeroing-out every $X_j,Y_j,$ and $Z_j$ such that $j \notin I$; we claim this is the desired zeroing out. Evidently $T_{iii}$ is not zeroed out for any $i\in I$. Meanwhile, for any other $(i,j,k) \in S$, it must be that $T_{ijk}$ was zeroed out, i.e., at least one of $i,j,k$ is not in $I$, since either at least one of $i,j,k$ is not in $R$, in which case it would not be included in $I$, or else $i$ would have been included in $A$ and hence excluded from $I$.
\end{proof}

Theorem~\ref{thm:newidea} shows how to start with a tensor $T$ which is a direct sum $\bigoplus_{i=1}^n T_{iii}$ plus roughly $m \cdot n$ additional subtensors $T_{ijk}$, and zero-out some variables so that $\Omega(n / \sqrt{m})$ of the $T_{iii}$ tensors remain, but all other subtensors are zeroed out. We will use this fact in the proof of Theorem~\ref{thm:main} below to show that the \emph{value} of $T$ is at least $V_\tau(T) \geq \Omega\left(\frac{n}{\sqrt{m}} \cdot \min_{i \in [n]} V_\tau(T_{iii})\right)$. It is natural to ask whether this $\sqrt{m}$ dependence is optimal; we next construct some tensors for which it is.

We first give a technical ingredient, Theorem~\ref{thm:opt1}, which we will use in our tensor construction.

\begin{theorem} \label{thm:opt1}
For positive integers $n \geq m$ with $n$ sufficiently large, there is a subset $S \subseteq \{ (i,j,k) \in [n]^3 \mid i,j,k \text{ distinct}\}$ with $|S|=mn$ such that, for any subset $I \subseteq [n]$ of size $|I| = \frac{n \log n}{\sqrt{m}}$, there is an $(i,j,k) \in S$ with $i,j,k \in I$. 
\end{theorem}

\begin{proof}
Let $C$ be the set of subsets $I \subseteq [n]$ of size $\frac{n \log n}{\sqrt{m}}$. Hence, $$|C| = \binom{n}{\frac{n \log n}{\sqrt{m}}} \leq n^{\frac{n \log n}{\sqrt{m}}}.$$ Initially let $S = \emptyset$. We will repeatedly add an element $(i,j,k)$ to $S$, and then remove any remaining $I \in C$ with $i,j,k \in I$, until $C$ becomes empty. It suffices to show we only need to add $\leq mn$ elements to $S$.
 
At each step, we simply greedily pick any $(i,j,k) \in [n]^3$ with $i,j,k$ distinct which maximizes $|\{ I \in C \mid i,j,k \in I \}|$. Note that if we pick three random distinct $i,j,k \in [n]$, then for a given $I \in C$, the probability that $i,j,k \in I$ is $$\frac{|I|}{n} \cdot \frac{|I|-1}{n-1} \cdot \frac{|I|-2}{n-2} > \left( \frac{|I|-2}{n-2} \right)^3 = \left( \frac{\frac{n \log n}{\sqrt{m}}-2}{n-2} \right)^3 > \frac12 \left( \frac{\log n}{\sqrt{m}} \right)^3$$ for large enough $n$. It follows that we can pick $i,j,k$ which multiply $|C|$ by a factor which is less than $1 - \frac12 \left( \frac{ \log n}{\sqrt{m}} \right)^3$. After repeating $nm$ times, the resulting size of $C$ will be less than 
$$n^{\frac{n}{\sqrt{m}}} \cdot \left( 1 - \frac12 \left( \frac{\log n}{\sqrt{m}} \right)^3 \right)^{mn} < n^{\frac{n}{\sqrt{m}}} \cdot \left( \frac{1}{e} \right)^{\frac12 n \log^3 n / \sqrt{m}} < 1,$$ for large enough $n$. Since $|C|$ is an integer, it follows that $|C|=0$ as desired.
\end{proof}

Let $S \subseteq [n]^3$ be the set from Theorem~\ref{thm:opt1}, with $|S| = m \cdot n$, and define the tensor $T$ over $\{x_1, \ldots, x_n\}$, $\{y_1, \ldots, y_n\}$, $\{z_1, \ldots, z_n\}$ by $$T = \sum_{(i,j,k) \in S} x_i y_j z_k.$$ Theorem~\ref{thm:opt1} says that, for any $I \subseteq [n]$ such that $T$ has a zeroing out into $\sum_{i \in I} x_iy_iz_i$, we must have $|I| < \frac{n \log n}{\sqrt{m}}$. This is nearly the size of the set $I$ constructed by Theorem~\ref{thm:newidea}.

The tensors to which we will apply Theorem~\ref{thm:newidea} will have additional structure beyond those stipulated by Theorem~\ref{thm:newidea}. It is worth investigating whether the $\sqrt{m}$ factor in Theorem~\ref{thm:newidea} can be improved for those tensors in particular. One particular property is that they are \emph{free}, meaning, for any $i,j,k,i',j',k' \in [n]$, such that $x_i y_j z_k$ and $x_{i'}y_{j'}z_{k'}$ both have nonzero coefficients in $T$, at most one of $i=i'$,$j'=j$,$k=k'$ holds. We can see that the tensor constructed from Theorem~\ref{thm:opt1} is unlikely to be free. Nonetheless, with some additional work, we can construct a free tensor for which the $\sqrt{m}$ factor is still optimal:

\begin{theorem} \label{thm:opt2}
For positive integers $n$ and $m$, with $n$ sufficiently large and $\log^2 n \leq m\leq \sqrt{n/6}$, there is a subset $S \subseteq \binom{[n]}{3}$ with $|S|\leq O(mn)$ such that
\begin{enumerate}[(1)]
    \item any distinct $\{i,j,k\},\{i',j',k'\} \in S$, have $|\{i,j,k\} \cap \{i',j',k'\}| \leq 1$, and
    \item for any subset $I \subseteq [n]$ of size $|I| = \frac{n \log (n)}{\sqrt{m}}$, there is an $\{i,j,k\} \in S$ with $i,j,k \in I$. 
\end{enumerate}
\end{theorem}

\begin{proof}
Let $N=2n$, and let $T = \binom{[N]}{3}$. Construct $S \subseteq T$ randomly by including each element of $T$ independently with probability $p := \frac{Nm}{3|T|}$. 

We compute some probabilities related to $S$.

First, note that $\E[|S|] = p \cdot |T| = \frac{Nm}{3}$, and so by Markov's inequality, we have $|S| \leq Nm=2nm$ with probability $\geq 2/3$.

Second, consider any fixed set $I \subseteq [N]$ of size $|I| = \frac{n \log n}{\sqrt{m}}$. 
For any three fixed elements $i,j,k\in I$ which are pairwise distinct, the probability that $\{i,j,k\}$ is not in $S$ is $1-\frac{Nm}{3|T|}\leq 1-\frac{2m}{N^2}$. Thus, for large enough $n$, the probability that no triple of $I$ is in $S$ is at most $$\left(1-\frac{2m}{N^2}\right)^{N^3(\log^3 N)/(16m^{3/2})}\leq 2^{-\Omega(\frac{N\log^3 N}{\sqrt{m}})}.$$

Meanwhile, the number of sets  $I \subseteq [N]$ of size $|I| = \frac{n \log n}{\sqrt{m}}$ is only
$$\binom{2n}{\frac{n \log n}{\sqrt{m}}} \leq  2^{O(\frac{N \log^2 N}{\sqrt{m}})}.$$
Thus the probability that all of those sets $I$ are covered by a triple in $S$ is overwhelming.

Let $U=[N]$ be the original universe.
Now, repeat the following procedure that shrinks $U$ somewhat. If there are two triples $\{i,j,k\},\{i,j,k'\} \in S$ which share two elements, then remove $i$ from $U$ (decreasing the size of $U$ by one, e.g. effectively making $U$ into $[N-1]$ the first time an element is removed).
This removes all subsets $I$ of size $n(\log n)/\sqrt m$ containing $i$ and also removes all triples of $S$ containing $i$.
The remaining subsets $I$ of $U$ of size $n(\log n)/\sqrt m$ are still covered by the remaining triples of $S$ as long as they were before we shrank $U$.

After this greedy procedure there are no more pairs of triples in $S$ that share a pair of elements.
Let us consider the size of $U$ after the greedy procedure.

Let us fix two triples $(i,j,k),(i,j,k')\in T$. The probability that both of them end up in (the original) $S$ is $\frac{N^2 m^2}{9|T|^2}\leq \frac{2m^2}{N^4}$. The number of pairs of triples that share a pair of elements is $\leq N^4$. Thus the expected number of such pairs that end up in $T$ is $\leq 2m^2$.
By Markov's inequality, the probability that there are $> 6m^2$ such pairs in $T$ is $\leq 1/3$. 

Thus, with probability at least $1-1/3-1/3 = 1/3$, the original $S$ had size $\leq mN=2mn$ and we removed $\leq 6m^2$ elements from the universe. Since $m\leq \sqrt{n/6}$, we have removed $\leq N/2$ elements, and so the remaining universe size is at lest $N/2=n$. 
If necessary, remove more elements from $U$ until $|U|=n$, effectively removing the triples of $S$ and subsets of $U$ of size $n(\log n)/\sqrt m$ that contain these elements.

We get that for the remaining $S$,  $|S|\leq O(mn)$, and with high probability, all subsets of $U$ of size $(n/\sqrt{m})\log (n)$ are covered by $S$.
\end{proof}

The requirement in Theorem~\ref{thm:opt2} that $m \leq \sqrt{n/6}$ may seem restrictive, but all the tensors to which we will apply Theorem~\ref{thm:main} have this property. Of course, the tensors to which we will apply Theorem~\ref{thm:newidea} have even more structure still than just being free. We leave open the question of whether Theorem~\ref{thm:newidea} can be further improved for them.

\section{Refined Laser Method} \label{sec:proof}

Let $T$ be a tensor over $X,Y,Z$, with partitions $X = X_1 \cup X_2 \cup \cdots \cup X_{k_X}$, $Y = Y_1 \cup Y_2 \cup \cdots \cup Y_{k_Y}$, $Z = Z_1 \cup Z_2 \cup \cdots \cup Z_{k_Z}$ for some positive integers $k_X,k_Y,k_Z$, and for $(i,j,k) \in [k_X]\times[k_Y]\times[k_Z]$ define $T_{ijk} := T|_{X_i,Y_j,Z_k}$. Let $S := \{ (i,j,k) \in [k_X]\times[k_Y]\times[k_Z] \mid T_{ijk} \neq 0\}$, and suppose there is an integer $P$ such that every $(i,j,k) \in S$ satisfies $i+j+k=P$. We call $T$ along with these partitions a \emph{$P$-partitioned tensor}. Some prior work called $T$ a partitioned tensor whose \emph{outer structure} is the tensor $$\sum_{i \in [k_X], j\in [k_Y], k \in [k_Z], i+j+k=P} x_iy_jz_k.$$

Let $D$ be the set of $\alpha : S \to [0,1]$  such that $\sum_{(i,j,k) \in S} \alpha(i,j,k) = 1$. For each $\alpha\in D$, we define a few quantities. 

First, for $(i,j,k) \in S$, write $\alpha_{ijk} := \alpha(i,j,k)$. 
For $i \in [k_X]$ write $$\alpha_{X_i} := \sum_{j \in [k_Y], k \in [k_Z] \mid (i,j,k) \in S} \alpha_{ijk},$$ and similarly define $\alpha_{Y_j}$ for $j \in [k_Y]$ and $\alpha_{Z_k}$ for $k \in [k_Z]$. Define the three products\footnote{We use the convention $0^0 := 1$.}

$$\alpha_B := \left( \prod_{i \in [k_X]} \alpha_{X_i}^{-\alpha_{X_i}} \right)^{1/3} \cdot \left( \prod_{j \in [k_Y]} \alpha_{Y_j}^{-\alpha_{Y_j}} \right)^{1/3} \cdot \left( \prod_{k \in [k_Z]} \alpha_{Z_k}^{-\alpha_{Z_k}} \right)^{1/3},$$

$$\alpha_N := \prod_{(i,j,k) \in S} \alpha_{ijk}^{-\alpha_{ijk}} , \text{ and}$$
$$\alpha_{V_\tau} := \prod_{(i,j,k) \in S} V_\tau(T_{ijk})^{\alpha_{ijk}} \text{ for } \tau \in [2/3,1].$$

Finally, define $D_\alpha \subseteq D$, the set of $\beta \in D$ which have the same marginals as $\alpha$, by $$D_\alpha := \{ \beta \in D \mid \alpha_{X_i}=\beta_{X_i} ~\forall i \in [k_X],~ \alpha_{Y_j}=\beta_{Y_j} ~\forall j\in [k_Y],~ \alpha_{Z_k}=\beta_{Z_k} ~\forall k \in [k_Z] \}.$$

We will show how to get a lower bound on $V_\tau(T)$ in terms of a given $\alpha \in D$, as follows:
\begin{theorem}[Refined Laser Method] \label{thm:main}
For any tensor $T$ which is a $P$-partitioned tensor, any $\alpha \in D$, and any $\tau \in [2/3,1]$, we have
$$V_\tau(T) \geq \alpha_{V_\tau} \cdot \alpha_B \cdot \sqrt{\frac{\alpha_N}{\max_{\beta \in D_\alpha} \beta_N}}.$$
\end{theorem}

By comparison, the bound used by prior work~\cite{v12,legall} was 
$$V_\tau(T) \geq \alpha_{V_\tau} \cdot \alpha_B \cdot \frac{\alpha_N}{\max_{\beta \in D_\alpha} \beta_N}.$$ Our Theorem~\ref{thm:main} improves this by a factor of $\sqrt{\frac{\max_{\beta \in D_\alpha} \beta_N}{\alpha_N}}$, which is a strict improvement whenever there is a $\beta \in D_\alpha$ with $\beta_N > \alpha_N$. As we will see, this is frequently the case in the analysis of powers of $CW_q$ and their subtensors.

Throughout this section, we omit $\tau$ when writing $V_\tau$ and $\alpha_{V_\tau}$, and we will always use the specific $\tau$ in the statement of Theorem~\ref{thm:main}.

\subsection{Proof Plan}
In the remainder of this section, we prove Theorem~\ref{thm:main}. Our proof strategy is as follows. Pick a large positive integer $n$, 
and consider the tensor $\T := T^{\otimes n} \otimes T^{r \otimes n} \otimes T^{rr \otimes n}$. We are going to show that $\T$ can be zeroed out into a direct sum of $$\frac{\alpha_B^{3n - o(n)} \cdot \alpha_N^{1.5n - o(n)}}{\max_{\beta \in D_\alpha} \beta_N^{1.5 n - o(n)}}$$ different tensors, each of which has value $$\alpha_V^{3n},$$ which will imply the bound \begin{align}\label{eq:valboundn}V_\tau(\T) \geq \frac{\alpha_V^{3n} \cdot \alpha_B^{3n - o(n)} \cdot \alpha_N^{1.5n - o(n)}}{\max_{\beta \in D_\alpha} \beta_N^{1.5 n - o(n)}}.\end{align} As $n \to \infty$, this implies the desired bound on $V_\tau(T) = V_\tau(\T)^{1/3n}$.

Our construction is divided into four main steps. The first three are mostly the same as the laser method from past work~\cite{coppersmith,stothers,v12,legall}, except that our analysis in step 3, in which we make use of Salem-Spencer sets, is more involved than in past work as we choose different parameters and need to preserve different properties of our tensor than in previous uses of the laser method. The main novel idea comes in step 4, where we apply our Theorem~\ref{thm:newidea} as a final zeroing-out step which has not appeared in past work.

Before we begin, we make one technical remark: we will assume throughout this proof that $\alpha_{ijk} \cdot n$ is an integer for all $(i,j,k) \in S$. This can be achieved by adding a real number of magnitude at most $1/n$ to each $\alpha_{ijk}$ so that they are all integer multiples of $1/n$, while maintaining that $\sum_{(i,j,k) \in S} \alpha(i,j,k) = 1$. In Appendix~\ref{app:rounding} below, we show that this is possible, and that the resulting changes to $\alpha$ only change the value bound (\ref{eq:valboundn}) by a negligible multiplicative factor between $2^{-o(n)}$ and $2^{o(n)}$. Hence, this `rounding' will not change the final bound in our proof.

\subsection{Step 1: Removing blocks which are inconsistent with \texorpdfstring{$\alpha$}{alpha}}

Let $n$ and $\T := T^{\otimes n} \otimes T^{r \otimes n} \otimes T^{rr \otimes n}$ be as above, and note that $\T$ is a tensor over $\X := X^n \times Y^n \times Z^n, \Y := Y^n \times Z^n \times X^n, \Z := Z^n \times X^n \times Y^n$. For $I \in [k_X]^n$, we write $X_I := \prod_{\ell \in [n]} X_{I_\ell} \subseteq X^n$, so that $X^n$ is partitioned by the $X_I$ for $I \in [k_X]^n$. We similarly define $Y_J$ for $J \in [k_Y]^n$ and $Z_K$ for $K \in [k_Z]^n$. Thus, $\X$ is partitioned by $X_I \times Y_J \times Z_K$ for $(I,J,K) \in [k_X]^n \times [k_Y]^n \times [k_Z]^n$; we call such an $X_I \times Y_J \times Z_K$ an \emph{$\X$-block}, and we similarly define $\Y$-blocks and $\Z$-blocks.

We say $X_I$ for $I \in [k_X]^n$ is \emph{consistent with $\alpha$} if, for all $i \in [k_X]$, we have $$\alpha_{X_i} = \frac{1}{n} \cdot | \{ \ell \in [n] \mid I_\ell = i \} |.$$ We define consistency with $\alpha$ for $Y_J$ for $J \in [k_Y]^n$, and $Z_K$ for $K \in [k_Z]^n$, similarly.

In $\T$, zero-out all $\X$-blocks $X_I \times Y_J \times Z_K$ where at least one of $X_I,Y_J,$ or $Z_K$ is not consistent with $\alpha$. Similarly, zero-out all $\Y$-blocks $Y_J \times Z_K \times X_I$ and all $\Z$-blocks $Z_K \times X_I \times Y_J$ where at least one of $X_I,Y_J,$ or $Z_K$ is not consistent with $\alpha$. Let $\T'$ denote $\T$ after these zeroing outs.

The number of $X_I$ which are consistent with $\alpha$ is $$\binom{n}{[\alpha_{X_i}\cdot n]_{i \in [k_X]}} = \left( \prod_{i \in [k_X]} \alpha_{X_i}^{-\alpha_{X_i}} \right)^{n - o(n)}.$$
Hence, recalling that $$\alpha_B = \left( \prod_{i \in [k_X]} \alpha_{X_i}^{-\alpha_{X_i}} \right)^{1/3} \cdot \left( \prod_{j \in [k_Y]} \alpha_{Y_j}^{-\alpha_{Y_j}} \right)^{1/3} \cdot \left( \prod_{k \in [k_Z]} \alpha_{Z_k}^{-\alpha_{Z_k}} \right)^{1/3},$$ we have that the number $N_B$ of remaining (not zeroed out) $\X$-blocks, $\Y$-blocks, or $\Z$-blocks in $\T'$ is $$N_B = \alpha_B^{3n - o(n)}.$$

\subsection{Step 2: Defining and counting nonzero block triples in \texorpdfstring{$\T'$}{T'}}

For an $\X$-block $B_X = X_I \times Y_{J'} \times Z_{K''}$, $\Y$-block $B_Y = Y_{J} \times Z_{K'} \times X_{I''}$, and $\Z$-block $B_Z = Z_{K} \times X_{I'} \times Y_{J''}$, we write $\T'_{B_X B_Y B_Z} := \T'|_{B_X,B_Y,B_Z}$. We call $\T'_{B_X B_Y B_Z} $ a \emph{block triple}, and say that it \emph{uses} $B_X,B_Y$, and $B_Z$. For a $\beta \in D$, we say that $(X_I,Y_J,Z_K)$ is \emph{consistent with $\beta$} if, for all $(i,j,k) \in S$,
$$|\{ \ell \in [n] \mid (I_\ell,J_\ell,K_\ell)=(i,j,k) \}| = \beta_{ijk} \cdot n.$$
Similarly, for $\beta, \beta', \beta'' \in D$, we say that $\T'_{B_X B_Y B_Z}$ is consistent with $(\beta,\beta',\beta'')$ if $(X_I,Y_J,Z_K)$ is consistent with $\beta$, $(X_{I'},Y_{J'},Z_{K'})$ is consistent with $\beta'$, and $(X_{I''},Y_{J''},Z_{K''})$ is consistent with $\beta''$.

Let $D_{\alpha,n} \subseteq D_\alpha$ be the set of $\beta \in D_\alpha$ such that $\beta_{ijk}$ is an integer multiple of $1/n$ for all $(i,j,k) \in S$.
Note that, because of the zeroing outs from the previous step, every nonzero block triple $\T'_{B_X B_Y B_Z}$ in $\T'$ is consistent with a $(\beta,\beta',\beta'')$ where $\beta,\beta',\beta'' \in D_{\alpha,n}$.
We can count that $|D_{\alpha,n}| \leq \poly(n)$, for some polynomial depending only on $k_X, k_Y, k_Z$.

Recalling that $$\beta_N = \prod_{(i,j,k) \in S} \beta_{ijk}^{-\beta_{ijk}},$$ we can count that the number of nonzero block triples $\T'_{B_X B_Y B_Z}$ in $\T'$ consistent with a given $(\beta,\beta',\beta'')$ for $\beta,\beta',\beta'' \in D_{\alpha,n}$ is  \begin{align*}&\binom{n}{[\beta_{ijk} \cdot n]_{(i,j,k) \in S}} \cdot \binom{n}{[{\beta'}_{ijk} \cdot n]_{(i,j,k) \in S}} \cdot \binom{n}{[{\beta''}_{ijk} \cdot n]_{(i,j,k) \in S}} \\ = & \left( \prod_{(i,j,k) \in S} \beta_{ijk}^{-\beta_{ijk}} \right)^{n-o(n)} \cdot \left( \prod_{(i,j,k) \in S} {\beta'}_{ijk}^{-{\beta'}_{ijk}} \right)^{n-o(n)} \cdot \left( \prod_{(i,j,k) \in S} {\beta''}_{ijk}^{-{\beta''}_{ijk}} \right)^{n-o(n)} \\ = &(\beta_N \cdot \beta'_N \cdot \beta''_N)^{n - o(n)}.\end{align*}

In particular, the number $N_\alpha$ of nonzero block triples in $\T'$ consistent with $(\alpha,\alpha,\alpha)$ is $$N_\alpha = \alpha_N^{3n - o(n)}.$$

Moreover, we can count that the total number $N_T$ of nonzero block triples in $\T'$ is at most
\begin{align*}
    N_T \leq \sum_{\beta, \beta', \beta'' \in D_{\alpha,n}} (\beta_N \cdot \beta'_N \cdot \beta''_N)^{n - o(n)} \leq |D_{\alpha,n}|^3 \cdot \max_{\beta \in D_{\alpha,n}} \beta_N^{3n - o(n)} \leq \poly(n) \cdot \max_{\beta \in D_{\alpha}} \beta_N^{3n - o(n)}.
\end{align*}

\subsection{Step 3: Carefully sparsifying \texorpdfstring{$\T'$}{T'} using Salem-Spencer sets}\label{subsec:step3}

We will now describe a randomized process for zeroing-out more $\X$-blocks, $\Y$-blocks, and $\Z$-blocks. The ultimate goal is to zero-out blocks so that at least a $\exp(-o(n))$ fraction of blocks is not zeroed out, and so that for every remaining $\X$-block, $\Y$-block, and $\Z$-block, there is exactly one block triple which uses that block and is consistent with $(\alpha,\alpha,\alpha)$. Note that currently, by symmetry, every  $\X$-block, $\Y$-block, and $\Z$-block in $\T'$ has the same number $R$ of block triples which use that block and are consistent with $(\alpha,\alpha,\alpha)$. We can compute $R$ by dividing the total number of block triples consistent with $(\alpha,\alpha,\alpha)$ by the total number of blocks, to see that $$R = \frac{N_\alpha}{N_B}.$$

Let $M$ be a prime number in the range $[100R,200R]$. We are going to define three random hash functions $h_X : [k_X]^n \times [k_Y]^n \times [k_Z]^n \to \mathbb{Z}_M$, $h_Y : [k_Y]^n \times [k_Z]^n \times [k_X]^n \to \mathbb{Z}_M$, and $h_Z : [k_Z]^n \times [k_X]^n \times [k_Y]^n \to \mathbb{Z}_M$ as follows, similar to~\cite{strassenlaser1,coppersmith} and subsequent work.
Recall that there is an integer $P$ such that every $(i,j,k) \in S$ satisfies $i+j+k=P$.
Pick independently and uniformly random $w_0, w_1, w_2, \ldots, w_{3n} \in \mathbb{Z}_M$. Define $h_X, h_Y, h_Z$, for $I \in [k_X]^n,J \in [k_Y]^n, K \in [k_Z]^n$, by:

\begin{align*} h_X (I,J,K) &:= 2\sum_{\ell=1}^n (w_\ell \cdot I_\ell + w_{\ell+n} \cdot J_\ell + w_{\ell+2n} \cdot K_\ell) \ \pmod{M},\\
 h_Y (J,K,I) &:= 2w_0 + 2\sum_{\ell=1}^n (w_\ell \cdot J_\ell + w_{\ell+n} \cdot K_\ell + w_{\ell+2n} \cdot I_\ell) \pmod{M},\\
 h_Z (K,I,J) &:= w_0 + \sum_{\ell=1}^n (w_\ell \cdot (P - K_\ell) + w_{\ell+n} \cdot (P-I_\ell) + w_{\ell+2n} \cdot (P-J_\ell)) \pmod{M}.\end{align*}

Consider any $\X$-block $B_X = X_I \times Y_{J'} \times Z_{K''}$, $\Y$-block $B_Y = Y_{J} \times Z_{K'} \times X_{I''}$, and $\Z$-block $B_Z = Z_{K} \times X_{I'} \times Y_{J''}$, such that $\T'_{B_X B_Y B_Z}$ is a nonzero block triple in $\T'$. Notice that
\begin{itemize}
    \item $h_X(I,J',K'') + h_Y(J,K',I'') = 2h_Z(K,I',J'') \pmod{M}$, regardless of the choice of randomness, since for every $\ell \in [n]$ we have $I_\ell+J_\ell+K_\ell = I'_\ell+J'_\ell+K'_\ell = I''_\ell+J''_\ell+K''_\ell = P$,
    \item The three values $h_X(I,J',K'')$, $h_Y(J,K',I'')$, and $h_Z(K,I',J'')$ are each uniformly random values in $\mathbb{Z}_M$, even when conditioned on one of the other two values, and 
    \item $h_X$ is pairwise-independent, i.e., for any two distinct $(I,J,K), (I',J',K') \in [k_X]^n \times [k_Y]^n \times [k_Z]^n$, the values $h_X(I,J,K)$ and $h_X(I',J',K')$ are independent (and similarly for $h_Y$ or $h_Z$).
\end{itemize}
For an $\X$-block $B_X = X_I \times Y_{J'} \times Z_{K''}$, write $h_X(B_X) := h_X(I,J',K'')$, and similarly define $h_Y(B_Y)$ for a $\Y$-block $B_Y$ and $h_Z(B_Z)$ for a $\Z$-block $B_Z$.

Let $A \subseteq \mathbb{Z}_M$ be a set of size $|A| \geq M^{1-o(1)}$, such that for $a,b,c \in A$, we have $a+b=2c \pmod{M}$ if and only if $a=b=c$. Such a set $A$ exists with these properties by Theorem~\ref{thm:SS}. In $\T'$, zero-out every $\X$-block $B_X$ such that $h_X(B_X) \notin A$, every $\Y$-block $B_Y$ such that $h_Y(B_Y) \notin A$, and every $\Z$-block $B_Z$ such that $h_Z(B_Z) \notin A$. Let the resulting tensor be $\T''$. 

Previous iterations of the laser method used this same hashing scheme, but we now need to analyze it more carefully to determine what parameters it gives when we use it in conjunction with Theorem~\ref{thm:newidea} in step 4 below. Toward this goal, we will bound the expected values of the following three random variables:
\begin{itemize}
    \item $C_1$, the number of nonzero block triples in $\T''$ consistent with $(\alpha,\alpha,\alpha)$,
    \item $C_2$, the number of pairs of nonzero block triples in $\T''$ which are both consistent with $(\alpha,\alpha,\alpha)$ and which share a block (i.e., both use the same $\X$-block, $\Y$-block, or $\Z$-block), and
    \item $C_3$, the number of nonzero block triples in $\T''$.
\end{itemize}

We start with $C_1$. The number of nonzero block triples in $\T'$ consistent with $(\alpha,\alpha,\alpha)$ is $N_\alpha$. Each is not zeroed out in $\T''$ with probability $\frac{|A|}{M^2}$, since this happens if and only if its $\X$-block hashes to a value in $A$ (which happens with probability $\frac{|A|}{M}$) and its $\Y$-block hashes to that same value (which happens with probability $\frac{1}{M}$ since $h_Y(B_Y)$ is uniformly random, even when conditioned on $h_X(B_X)$). Hence, $\E[C_1] = \frac{|A| \cdot N_\alpha}{M^2}$.

We next consider $C_2$, the number of pairs of nonzero block triples in $\T''$ which are both consistent with $(\alpha,\alpha,\alpha)$ and which both use the same $\X$-block, $\Y$-block, or $\Z$-block. The number of such pairs in $\T'$ is $3N_B \cdot \binom{R}{2}$, since there are $N_B$ each of $\X$-blocks, $\Y$-blocks, and $\Z$-blocks in $\T'$, and each has $R$ different block triples consistent with $(\alpha,\alpha,\alpha)$ which use it. Consider a fixed one of those pairs of $\T'_{B_X, B_Y, B_Z}$ and $\T'_{B_X, B_Y', B_Z'}$, where we are assuming without loss of generality that they share an $\X$-block $B_X$. They will both not be zeroed out in $\T''$ if and only if $h_X(B_X) \in A$, $h_Y(B_Y) = h_X(B_X)$, and $h_Y(B_Y') = h_X(B_X)$, which happens with probability $(|A|/M) \cdot (1/M) \cdot (1/M) = |A|/M^{3}$, as those three events are independent from the properties of $h_X$ and $h_Y$. Recall that $R = N_\alpha / N_B$ by definition of $R$, and $M \geq 100\cdot R$ by definition of $M$. Hence, using linearity of expectation, we can bound $$\E[C_2] = \frac{|A|}{M^3} \cdot 3 N_B \cdot \binom{R}{2} \leq \frac{3 \cdot |A| \cdot N_B \cdot R^2}{2 M^3} \leq \frac{3 \cdot |A| \cdot N_B \cdot R}{200 M^2} = \frac{3 \cdot |A| \cdot N_\alpha}{200 M^2}.$$

Finally, we consider $C_3$, the number of nonzero block triples in $\T''$. Similar to $C_1$, each of the $N_T$ nonzero block triples in $\T'$ is not zeroed out in $\T''$ with probability $\frac{|A|}{M^2}$, and so $\E[C_3] = \frac{|A| \cdot N_T}{M^2}$.

Now we define the random variable $C_1' := \max\{ 0, C_1 - 2 C_2 \}$. We have $$\E[C_1'] \geq \E[C_1 - 2 C_2] = \E[C_1] - 2 \E[C_2] \geq \frac{|A| \cdot N_\alpha}{M^2} - 2 \frac{3 \cdot |A| \cdot N_\alpha}{200 M^2} = \frac{97 \cdot |A| \cdot N_\alpha}{100 \cdot M^2}.$$

Since $\E[C_1'] \geq \frac{97 \cdot |A| \cdot N_\alpha}{100 \cdot M^2}$ and $\E[C_3] = \frac{|A| \cdot N_T}{M^2}$, it follows that there is a choice of randomness (i.e., a choice of $w_0, w_1, w_2, \ldots, w_{3n}$ defining the hash functions) for which \begin{align} \frac{C_1'^{3/2}}{C_3^{1/2}} \geq \frac{\left(\frac{97 \cdot |A| \cdot N_\alpha}{100 \cdot M^2} \right)^{3/2}}{\left( \frac{|A| \cdot N_T}{M^2} \right)^{1/2}}. \end{align} (This follows from the power mean inequality; see Lemma~\ref{lem:apppowermean} in Appendix~\ref{app:powermean} below for a proof.) Let us fix this choice of randomness in the remainder of the proof.

\subsection{Step 4: Converting \texorpdfstring{$\T''$}{T''} to an independent sum of \texorpdfstring{$(\alpha,\alpha,\alpha)$}{(alpha,alpha,alpha)}-consistent block triples}

Next, we will zero-out some more blocks in $\T''$ so that there are no pairs of nonzero block triples in $\T''$ which are both consistent with $(\alpha,\alpha,\alpha)$ and which both use the same $\X$-block, $\Y$-block, or $\Z$-block. We do this in the following greedy way: repeatedly pick any block used by $g \geq 2$ nonzero block triples in $\T''$ consistent with $(\alpha,\alpha,\alpha)$, and zero-out that block, until there are none left. Note that each time, when we zero-out $g \geq 2$ block triples which are consistent with $(\alpha,\alpha,\alpha)$, the number of \emph{pairs} of such block triples we remove is $\binom{g}{2} \geq g/2$. Recall that there were initially $C_2$ such pairs. It follows that throughout this process, the expected number of such block triples we zero-out is at most $2 \cdot C_2$. Let $\T'''$ be $\T''$ after this process is complete. Hence, the remaining nonzero block triples in $\T'''$ consistent with $(\alpha,\alpha,\alpha)$ do not share any blocks with each other, and the number of such block triples is at least $\max\{0, C_1 - 2 \cdot C_2\} = C_1'$.
Meanwhile, the total number of nonzero block triples in $\T'''$ is at most the number in $\T''$, which is $C_3$. 

For the final zeroing out, we will apply Theorem~\ref{thm:newidea} to $\T'''$. The partition of the variables of $\T'''$ needed for Theorem~\ref{thm:newidea} is given by the blocks, and we order them such that the block triples consistent with $(\alpha,\alpha,\alpha)$ appear along the diagonal; we know from the previous paragraph that they do not share blocks with each other. Recall that $\T'''$ consists of at least $C_1'$ block triples consistent with $(\alpha,\alpha,\alpha)$, plus at most $C_3$ other block triples. It follows from Theorem~\ref{thm:newidea} that we can zero-out $\T'''$ into a direct sum of $L$ block triples consistent with $(\alpha,\alpha,\alpha)$, where $$L \geq\frac{2}{3 \sqrt{3}} \cdot \frac{C_1'}{\sqrt{C_3/C_1'}}.$$

Recalling that $N_\alpha = \alpha_N^{3n - o(n)}$, $N_B = N_\alpha / R = \alpha_B^{3n - o(n)}$, and $N_T \leq \poly(n) \cdot \max_{\beta \in D_\alpha} \beta_N^{3 n - o(n)}$, we can lower-bound $L$ by
\allowdisplaybreaks \begin{align*}L &\geq\frac{2}{3 \sqrt{3}} \cdot \frac{C_1'}{\sqrt{C_3/C_1'}} 
\\ &= \frac{2}{3 \sqrt{3}} \cdot \frac{C_1'^{3/2}}{C_3^{1/2}}
\\ &\geq \frac{2}{3 \sqrt{3}} \cdot \frac{\left(\frac{97 \cdot |A|\cdot N_\alpha}{100 \cdot M^{2}}\right)^{3/2}}{\left(\frac{|A| \cdot N_T}{M^2}\right)^{1/2}}
\\ &= \frac{97 \sqrt{291}}{4500} \cdot \left(\frac{|A|\cdot N_\alpha}{M^{2}}\right) \cdot \sqrt{\frac{N_\alpha}{N_T}}
\\ &= \left(\frac{N_\alpha}{M^{1+o(1)}}\right) \cdot \sqrt{\frac{N_\alpha}{N_T}}
\\ &\geq  \left(\frac{N_\alpha}{(200\cdot R)^{1+o(1)}}\right) \cdot \sqrt{\frac{N_\alpha}{N_T}}
\\ &\geq  N_B^{1 - o(1)} \cdot \sqrt{\frac{N_\alpha^{1-o(1)}}{N_T}}
\\ &\geq \frac{1}{\poly(n)} \cdot  \frac{\alpha_B^{3n - o(n)} \cdot \alpha_N^{1.5n - o(n)}}{\max_{\beta \in D_\alpha} \beta_N^{1.5 n - o(n)}}.
\end{align*}
Each block triple $\T'''_{X_B, Y_B, Z_B}$ consistent with $(\alpha,\alpha,\alpha)$ can be written as $$\T'''_{X_B, Y_B, Z_B} \equiv \bigotimes_{(i,j,k) \in S} (T_{ijk} \otimes T_{ijk}^r \otimes T_{ijk}^{rr})^{\alpha_{ijk} \cdot n},$$
and hence has value $$V_\tau(\T'''_{X_B, Y_B, Z_B}) \geq \prod_{(i,j,k) \in S} V_\tau(T_{ijk} \otimes T_{ijk}^r \otimes T_{ijk}^{rr})^{\alpha_{ijk} \cdot n} = \prod_{(i,j,k) \in S} V_\tau(T_{ijk})^{3 \alpha_{ijk} \cdot n} = \alpha_V^{3n}.$$
It follows that
$$V_\tau(\T''') \geq L \cdot \alpha_V^{3n} \geq \frac{1}{\poly(n)} \cdot  \frac{\alpha_V^{3n} \cdot \alpha_B^{3n - o(n)} \cdot \alpha_N^{1.5n - o(n)}}{\max_{\beta \in D_\alpha} \beta_N^{1.5 n - o(n)}},$$
and hence that
$$V_\tau(T) \geq V_\tau(\T)^{\frac{1}{3n}} \geq V_\tau(\T''')^{\frac{1}{3n}} \geq \frac{\alpha_V \cdot \alpha_B^{1 - o(1)} \cdot \alpha_N^{1/2 - o(1)}}{\max_{\beta \in D_\alpha} \beta_N^{1/2 - o(1)}}.$$
The desired result follows as $n \to \infty$.

\section{Algorithms and Heuristics for Applying Theorem~\ref{thm:main}} \label{sec:heuristics}

We now move on to applying Theorem~\ref{thm:main} to the Coppersmith-Winograd tensor $CW_q$. Throughout this section we use the same notation as in Section~\ref{sec:proof}. The high-level idea is to apply Theorem~\ref{thm:main} in a recursive fashion: to bound $V_\tau(T)$ for a tensor $T$ (in our case, $T$ will be $CW_q^{\otimes k}$ or one of its subtensors), we pick a partitioning of its variables, recursively bound $V_\tau(T_{ijk})$ for each subtensor $T_{ijk}$, then pick an $\alpha \in D$ to get a resulting lower bound on $V_\tau(T)$.

When using this approach, there are two choices we need to make at each level: which partitioning of the variables to use, and which $\alpha \in D$ to pick. As we will see in Section~\ref{sec:final-answer}, there are very natural partitionings of the variables of $CW_q^{\otimes k}$ and its subtensors that we will use. The main practical difficulty which arises in this approach is picking the optimal value of $\alpha$. Indeed, maximizing the value bound of Theorem~\ref{thm:main} over all $\alpha \in D$ is a non-convex optimization problem, and for large enough tensors (including $CW_q^{\otimes 8}$, as well as the subtensors of $CW_q^{\otimes k}$ for $k \geq 16$), modern software seems unable to solve it in a reasonable amount of time. (Modern software does actually solve the optimization problem for the subtensors of $CW_q^{\otimes k}$ for $k\leq 8$.)

Instead, as in past work~\cite{v12,legall}, we use some heuristics to find choices of $\alpha$ which we believe are close to optimal. The heuristics we use are different from those of past work, in order to take advantage of the new improvement in Theorem~\ref{thm:main}; we will see that the heuristics we use here achieve better value bounds than the heuristics of past work. In the remainder of this section, we give a detailed description of these heuristics that we use.

\subsection{Optimizing over \texorpdfstring{$\alpha \in D_{\gamma}$}{alpha in D gamma} for fixed \texorpdfstring{$\gamma \in D$}{gamma in D}} \label{subsec:optgamma}

Although optimizing the bound of Theorem~\ref{thm:main} over all $\alpha \in D$ appears difficult, we begin by remarking that, for any fixed $\gamma \in D$, optimizing over all $\alpha \in D_\gamma$ can be done by solving two programs:

\begin{enumerate}[label={\lsstyle\interextrabold\fontsize{12}{12}\selectfont PROBLEM \arabic*.},  align= left, itemsep=1em, labelindent = 0em,ref=\arabic*] 
    \item \label{prob1} Maximize $$\alpha_{V_\tau} \cdot \alpha_B \cdot \sqrt{\alpha_N}$$ subject to $\alpha \in D_\gamma$.
    \item \label{prob2} Maximize $$\beta_N$$ subject to $\beta \in D_\gamma$.
\end{enumerate}

Recall that $\alpha \in D_\gamma$ if and only if $\alpha_{X_i} = \gamma_{X_i}$ for all $i \in [k_X]$, $\alpha_{Y_j} = \gamma_{Y_j}$ for all $j \in [k_Y]$, and $\alpha_{Z_k} = \gamma_{Z_k}$ for all $k \in [k_Z]$. In other words, in both of these problems, the constraints are all linear constraints. Meanwhile, in both, the objective function is a concave function. Hence, we can efficiently solve these two problems using a convex optimization library to obtain $\alpha, \beta \in D_\gamma$. Theorem~\ref{thm:main} then implies the bound $$V_\tau(T) \geq \alpha_{V_\tau} \cdot \alpha_B \cdot \sqrt{\frac{\alpha_N}{\beta_N}},$$
and this is the best possible bound we can get using an $\alpha \in D_\gamma$.

\subsection{Heuristics for picking \texorpdfstring{$\gamma \in D$}{gamma in D}} \label{subsec:heuristics}

As discussed, it seems computationally difficult to find the optimal $\gamma \in D$ to use in the above approach. Instead, when analyzing a tensor $T$, we try the following heuristic choices of $\gamma$, and take the maximum value bound which results from any of them.

Our first heuristic is a slight improvement of that of \cite[{Algorithm~B}]{legall}.

\begin{enumerate}[label={\lsstyle\interextrabold\fontsize{12}{12}\selectfont HEURISTIC \arabic*.},  align= left, itemsep=1em, labelindent = 0em,ref=\arabic*]  
\item \label{heur:1} Use the $\gamma$ which maximizes $$\gamma_{V_\tau} \cdot \gamma_B$$ subject to $\gamma \in D$ and $\gamma = \argmax_{\gamma' \in D_\gamma} \gamma_N$.
\end{enumerate}
With this choice of $\gamma$, we know that when we compute the value bound of Section~\ref{subsec:optgamma}, we will pick $\beta = \gamma$, and so the final value we output will be at least $\gamma_{V_\tau} \cdot \gamma_B$, but possibly even greater if a better choice of $\alpha$ is found. By comparison, \cite[{Algorithm~B}]{legall} computes this same $\gamma$, but then outputs the value one gets from picking $\alpha=\beta=\gamma$ in Section~\ref{subsec:optgamma}.

As written, it is not evident that Heuristic~\ref{heur:1} can be computed much more quickly than the optimal choice of $\gamma \in D$.
However, prior work \cite[{Figure~1}]{v12}, \cite[{Proposition~4.1}]{legall} gives a way to define a set of nonlinear constraints NonLin$(\gamma)$ on $\gamma \in D$ which are satisfied if and only if $\gamma = \argmax_{\gamma' \in D_\gamma} \gamma_N$. (These can also be defined by analyzing the convex program in Problem~\ref{prob2} directly.) The number of constraints in NonLin$(\gamma)$ is small enough for many of the subtensors of $CW_q^{\otimes 16}$ that optimization software can still compute the $\gamma$ in Heuristic~\ref{heur:1} after replacing the constraint $\gamma = \argmax_{\gamma' \in D_\gamma} \gamma_N$ with NonLin$(\gamma)$. We find that this heuristic runs quickly enough and gives good value bounds for many of the subtensors of $CW_q^{\otimes 16}$.

Our remaining heuristics only involve solving convex programs over linear constraints to compute $\gamma$, and run quickly enough for all the tensors we need to analyze.

\begin{enumerate}[label={\lsstyle\interextrabold\fontsize{12}{12}\selectfont HEURISTIC \arabic*.},  align= left, itemsep=1em, labelindent = 0em,ref=\arabic*]  \setcounter{enumi}{1}
     \item \label{heur:2} Use the $\gamma$ which maximizes  $$\gamma_{V_\tau} \cdot \gamma_B$$ subject to $\gamma \in D$.
     \item\label{heur:3} Use the $\gamma$ which maximizes  $$\gamma_{V_\tau} \cdot \gamma_B + \lambda/\gamma_N$$ subject to $\gamma \in D$, for various nonnegative constant parameters $\lambda$. This generalizes Heuristic~\ref{heur:2}. 
     \item\label{heur:4} Use the $\gamma$ which maximizes  $$\gamma_{V_\tau} \cdot \gamma_B \cdot \sqrt{\gamma_N}$$ subject to $\gamma \in D$.
\end{enumerate}

Out of these heuristics, the best bounds were obtained for most tensors via Heuristic~\ref{heur:3} for various choices of $\lambda$ between $0$ (as in Problem~\ref{prob2}) and $10^7$. It is not immediately clear which choice of $\lambda$ in Heuristic~\ref{heur:3} is best, although experimentally, it seems that using larger values of $\lambda$ produces better results. We also tried a few other heuristics, including maximizing $\gamma_{V_\tau} \cdot \gamma_B / \sqrt{\gamma_N}$ or $1/\gamma_N$ over $\gamma \in D$, but these didn't yield the best value bounds for any tensors we analyzed.

\subsection{Faster Computation on Tensors with Symmetries}

We briefly note one additional technique which can be used to speed up the calculations needed when applying Theorem~\ref{thm:main} to a tensor $T$ which exhibits some symmetry (either the calculations to apply the Theorem optimally, or the heuristics described earlier in this section).  Suppose, for instance, that for all $(i,j,k) \in S$ we have $V_\tau(T_{ijk}) = V_\tau(T_{jik})$. As we will see below, this will be the case in nearly all our applications of Theorem~\ref{thm:main}. Then, in all of the optimization problems over $\alpha \in D$ (or similarly $\beta \in D$ or $\gamma \in D$) that we need to solve, we may assume without loss of generality that $\alpha_{ijk} = \alpha_{jki}$ for all $(i,j,k) \in S$. Prior work also used such symmetry considerations; see~\cite{stothers,v12,legall} for more details.

%\section{Bounding  $V_\tau(CW_q)$} \label{sec:final-answer}
%\section{Bounding \texorpdfstring{$V_\tau(CW_q)$}{the tau-value of Coppersmith-Winograd tensors}} \label{sec:final-answer}

\section{Bounding \texorpdfstring{$V_\tau(CW_q)$}{V(CWq)}} \label{sec:final-answer}
Recall the definition of the family of Coppersmith-Winograd tensors $CW_q$ parameterized by integer $q\geq 0$. $CW_q$ is a tensor over $\{x_0,\ldots,x_{q+1}\}$, $\{y_0, \ldots, y_{q+1}\}$, $\{z_0, \ldots, z_{q+1}\}$ given by $$CW_q := x_0 y_0 z_{q+1} + x_0 z_{q+1} y_0 + x_{q+1} y_0 z_0 + \sum_{i=1}^q (x_0 y_i z_i + x_i y_0 z_i + x_i y_i z_0).$$ 
Coppersmith and Winograd~\cite{coppersmith} showed that $\tilde{R}(CW_q) = q+2$. We will apply the approach from Section~\ref{sec:heuristics} to give a lower bound on $V_\tau(CW_q)$, and hence an upper bound on $\omega$.

\subsection{Variable Partition for \texorpdfstring{$CW_q$}{CWq}}

$CW_q$ has a natural partitioning of its variables into $X=X^0\cup X^1\cup X^2$, $Y=Y^0\cup Y^1\cup Y^2$, $Z=Z^0\cup Z^1\cup Z^2$, where $X^0=\{x_0\}$, $X^1=\{x_1,\ldots,x_q\}$, $X^2=\{x_{q+1}\}$, $Y^0=\{y_0\}$, $Y^1=\{y_1,\ldots,y_q\}$, $Y^2=\{y_{q+1}\}$, and $Z^0=\{z_0\}$, $Z^1=\{z_1,\ldots,z_q\}$, $Z^2=\{z_{q+1}\}$. 
For $i,j,k \in \{0,1,2\}$, we write $T_{ijk} := CW_q|_{X^i,Y^j,Z^k}$, and call these the \emph{subtensors} of $CW_q$. We see that the nonzero subtensors are: $T_{002}=x_0 y_0 z_{q+1}$, $T_{020}= x_0 z_{q+1} y_0$,  $T_{200}=x_{q+1} y_0 z_0$, $T_{011}=\sum_{i=1}^q x_0 y_i z_i$,  $T_{101}=\sum_{i=1}^q x_i y_0 z_i$, and $T_{110}=\sum_{i=1}^q x_i y_i z_0$. In particular, we see that $T_{ijk} \neq 0$ if and only if $i+j+k=2$, meaning
$$CW_q=\sum_{i,j,k\in \{0,1,2\} \mid i+j+k=2} T_{ijk}.$$
Hence, $CW_q$ is a $2$-partitioned tensor as defined in Section~\ref{sec:proof}, and so we can apply Theorem~\ref{thm:main} to it to bound its value $V_\tau(CW_q)$.

The subtensors of $CW_q$ are all isomorphic to matrix multiplication tensors, and so their values can all be computed by following the definition of $V_\tau$. The subtensors
$T_{002}$, $T_{020}$, and $T_{200}$ are all isomorphic to $\langle 1,1,1\rangle$ and have value $1$ for any $\tau \in [2/3,1]$. Meanwhile, $T_{011}$, $T_{101}$, and $T_{110}$ are isomorphic to $\langle 1,1,q\rangle$ or one of its two rotations, and thus have value $q^\tau$. We can thus apply Theorem~\ref{thm:main} to bound the value of $CW_q$.

%\subsection{Variable Partition for $CW_q^{\otimes t}$}
\subsection{Variable Partition for \texorpdfstring{$CW_q^{\otimes t}$}{CWqot}}

However, as in prior work, instead of applying Theorem~\ref{thm:main} only to $CW_q$, we will apply it to $CW_q^{\otimes t}$ for $t$ a power of $2$. We focus on $t \in \{2,4,8,16,32\}$ as these are the powers for which the software solvers obtain solutions. In general, it is not clear that applying Theorem~\ref{thm:main} to a power of a tensor rather than the tensor itself should yield an improved value. However, prior work on analyzing the value of $CW_q$ has noticed that the laser method applied to powers of $CW_q$ can yield an improved bound because of a `merging' phenomenon, wherein we can prove better value bounds on certain subtensors of $CW_q^{\otimes t}$ by merging together different matrix multiplication tensors into single, larger matrix multiplication tensors. We will also take advantage of this here. 

We first describe the necessary variable partitioning of $CW_q^{\otimes t}$. Recall that $CW_q^{\otimes t}$ is a tensor over the variables $\{x_A\}$, $\{y_B\}$, $\{z_C\}$ for $A,B,C\in \{0,1,\ldots,q+1\}^{t}$. 
Inspired by the variable partitions of $CW_q$, we define a function $\kappa:\{0,1,\ldots,q+1\}\rightarrow \{0,1,2\}$ which maps $\kappa(0)=0$, $\kappa(q+1)=2$, and each other each $i\in \{1,\ldots,q\}$ to $\kappa(i)=1$. For $I \in \{0,1,\ldots,2t\}$, let $L_{t,I}$ be the set of index sequences whose coordinates sum to $I$, i.e., $L_{t,I} := \{ A \in \{0,1,\ldots,q+1\}^t \mid \sum_{\ell=1}^t \kappa(A_\ell)=I \}$. Then, for each $I,J,K \in \{0,\ldots,2t\}$, we define $X^{t,I} := \{x_A \mid A \in L_{t,I} \}$, $Y^{t,J} := \{y_B \mid B \in L_{t,J} \}$, and $Z^{t,K} := \{z_C \mid C \in L_{t,K} \}$, and write $T^{t}_{IJK} := CW_q^{\otimes t}|_{X^{t,I} Y^{t,J} Z^{t,K}}$.

As an example, the subtensor $T^2_{112}$ of $CW^{\otimes 2}_q$ is $$\sum_{i=1}^q x_{0,i} y_{0,i} z_{q+1,0} + \sum_{i=1}^q x_{i,0} y_{i,0} z_{0,q+1} + \sum_{i=1}^q\sum_{j=1}^q (x_{0,j} y_{i,0} z_{i,j} + x_{i,0} y_{0,j} z_{i,j}).$$

Since every term $x_iy_jz_k$ with nonzero coefficient in $CW_q$ has $\kappa(i)+\kappa(j)+\kappa(k)=2$, it follows that every nonzero $T^t_{IJK}$ in $CW_q^{\otimes t}$ has $I+J+K=2t$, and so the $X^{t,I},Y^{t,J},Z^{t,K}$ give a partitioning of the variables of $CW_q^{\otimes 2t}$ with
$$CW_q^{\otimes t} = \sum_{I,J,K\in \{0,\ldots,2t\} \mid I+J+K=2t} T^t_{IJK}.$$ Hence, $CW_q^{\otimes t}$ is a partitioned tensor with outer structure $C_{2t+1}$.

In order to apply Theorem~\ref{thm:main} to $CW_q^{\otimes t}$, we need a way to bound the values of subtensors $T^t_{IJK}$. As we can see with $T^2_{112}$ above, these subtensors are no longer always isomorphic to matrix multiplication tensors when $t \geq 2$, and so bounding these values will be less straightforward than before.

\subsection{Variable Partition for \texorpdfstring{$T^t_{IJK}$}{TIJKt}} \label{subsec:Tpart}
Let us assume that $t$ is even (as we will be dealing with powers of $2$).
Fix $I,J,K \in \{0,1,\ldots,2t\}$ with $I+J+K=2t$. To bound $V_\tau(T^t_{IJK})$, we will again use Theorem~\ref{thm:main}.  Again, to do this, we need a partition of the variables of $T^t_{IJK}$. The key remark is that for any $I \in \{0,1,\ldots,2t\}$ and any $A \in L_{t,I}$, there is some $I' \in \{0,1,\ldots,I\}$ such that $\sum_{\ell=1}^{t/2} A_\ell = I'$, and $\sum_{\ell=t/2+1}^t A_\ell = I-I'$. This splitting gives rise to the decomposition: 
$$T^{t}_{I,J,K} = \sum_{I',J',K'\in \{0,\ldots,t\}, I'+J'+K'=t, I'\leq I,J'\leq J, K'\leq K} T^{t/2}_{I',J',K'}\otimes T^{t/2}_{(I-I'),(J-J'),(K-K')}.$$

For instance, as can be seen above, $$T^2_{112}=T_{002}\otimes T_{110} + T_{110}\otimes T_{002} + T_{011}\otimes T_{101} + T_{101}\otimes T_{011}.$$

For $I',J',K' \in \{0,1,\ldots,t\}$ with $I'+J'+K'=t, I'\leq I,J'\leq J, K'\leq K$, write $$T^{t}(I',J',K'):=T^t_{I',J',K'}\otimes T^t_{I-I',J-J',K-K'}.$$
$T^{t}(I',J',K')$ is a tensor over the variables $X^{t/2,I'} \times X^{t/2,I-I'}$, $Y^{t/2,J'} \times Y^{t/2,J-J'}$, $Z^{t/2,K'} \times Z^{t/2,K-K'}$. Hence, these sets for $I',J',K' \in \{0,1,\ldots,t\}$ with $I'\leq I,J'\leq J, K'\leq K$ partition the variables of $T^{t}_{IJK}$, and they show that $T^{t}_{IJK}$ is a $t$-partitioned tensor via
$$T^{t}_{I,J,K} = \sum_{I',J',K'\in \{0,\ldots,t\}, I'+J'+K'=t, I'\leq I,J'\leq J, K'\leq K} T^{t}(I',J',K').$$

\subsection{Larger Values of Some Subtensors from Merging}

Finally, for a few subtensors, we can prove an even greater bound on their value than is given by the approach of Section~\ref{subsec:Tpart}. This is the key reason why analyzing higher powers of $CW_q$ with Theorem~\ref{thm:main} can yield higher values. The idea is to note that $T^t_{IJK}$ is isomorphic to a matrix multiplication whenever $I=0$, $J=0$, or $K=0$. 

Consider, for instance, $T^2_{220}$. As above, we have that
\begin{align*}
    T^2_{220} &= T_{200} \otimes T_{020} + T_{020} \otimes T_{200} + T_{110} \otimes T_{110} \\
    &= x_{q+1,0} y_{0,q+1} z_{0,0} + x_{0,q+1} y_{q+1,0} z_{0,0} + \sum_{i=1}^q x_{i,i} y_{i,i} z_{0,0}.
\end{align*}
If we `rename' $x_{q+1,0}$ to $x_{0,0}$, $y_{0,q+1}$ to $y_{0,0}$, $x_{0,q+1}$ to $x_{q+1,q+1}$, and $y_{q+1,0}$ to $y_{q+1,q+1}$, then this shows $$T^2_{220} \equiv x_{0,0} y_{0,0} z_{0,0} + x_{q+1,q+1} y_{q+1,q+1} z_{0,0} + \sum_{i=1}^q x_{i,i} y_{i,i} z_{0,0} = \sum_{i=0}^{q+1} x_{i,i} y_{i,i} z_{0,0} \equiv \langle q+2,1,1 \rangle.$$
Hence, $V_\tau(T^2_{220}) = (q+2)^\tau$. One can verify that this is better than we would have gotten by applying Theorem~\ref{thm:main} to $T^2_{220}$ as in Section~\ref{subsec:Tpart}. This is intuitively because that approach would have treated the three parts  $T_{200} \otimes T_{020}, T_{020} \otimes T_{200}, T_{110} \otimes T_{110}$ as separate tensors instead of `merging' them together into a single matrix multiplication tensor.

More generally, the subtensor $T^t_{IJ0}$ for $J \in \{0,1,\ldots,t/2\}$ and $I = t-J \geq J$ can be merged into a single matrix multiplication tensor, yielding the value $$V_\tau(T^t_{IJ0}) = \left( \sum_{b \leq J, b \equiv J\pmod{2}} \binom{t/2}{b,\frac{J-b}{2},\frac{I-b}{2}} \cdot q^b  \right)^\tau.$$
 See~\cite[{Claim~1}]{v12} for the full calculation. The similar value bound holds for any $T^t_{IJK}$ where at least one of $I,J,K$ is $0$ by symmetry.

\subsection{Numerical Value Bounds}

Finally, we have written code to carry out the recursive procedure described in this section. We ultimately find that, for $\tau = 2.3728596 / 3$, we get $V_\tau(CW_5^{\otimes 32}) > 7^{32} + 9.19 \times 10^{22}$. The code to verify this bound can be found at the link in~\cite{vercode}.
%~\url{http://code.joshalman.com/MM}. 
The basis of our code is the publicly available code of Le Gall~\cite{legall}. We then added new functions to compute the maximum possible value achievable by Theorem~\ref{thm:main}, as well as all of the aforementioned heuristics. Our code makes use of both the NLPSolve function in Maple~\cite{maple}, and the CVX convex optimization software for Matlab~\cite{cvx,gb08,MATLAB2019}.

The values for the subtensors of $CW_5^{\otimes t}$ for $t \in \{2,4,8\}$ are bounded by finding the optimal $\gamma$ to use in Section~\ref{subsec:optgamma}. Most of the values for the subtensors of $CW_5^{\otimes 16}$, and all of the values for the subtensors of $CW_5^{\otimes 32}$, are bounded using Heuristic~\ref{heur:2} from Section~\ref{subsec:heuristics}. For the subtensors of $CW_5^{\otimes 16}$ whose values are bounded using a different heuristic, the exact point $\gamma$ that we use in Section~\ref{subsec:optgamma} is included with the code\footnote{In fact, the $\gamma$ points we provide coincide with the $\beta$ points of Section~\ref{subsec:optgamma}.}. Finally, the `global' value bound for $CW_5^{\otimes 32}$ is computed using Heuristic~\ref{heur:2}. All of the subtensor value bounds we compute are provided along with the code.
 
\section*{Acknowledgements}

We would like to thank Fran{\c{c}}ois Le Gall for answering our questions about his code from~\cite{legall},  Ryan Williams for his suggestions for improving the presentation of this paper and optimizing our code, and anonymous reviewers for their helpful comments.

\printbibliography

\appendix 

\section{Rounding \texorpdfstring{$\alpha$}{alpha} in the proof of \texorpdfstring{Theorem~\ref{thm:main}}{Theorem 4.1}}\label{app:rounding}

Recall the notation from Section~\ref{sec:proof}. Fix any $\alpha \in D$ and any sufficiently large positive integer $n$. 

\begin{lemma}
There is an $\alpha' \in D$ such that, for all $(i,j,k) \in S$,
\begin{itemize}
    \item $\alpha'_{ijk}$ is an integer multiple of $\frac{1}{n}$, and
    \item $|\alpha_{ijk} - \alpha'_{ijk}| < \frac{1}{n}$.
\end{itemize}
\end{lemma}

\begin{proof}
Let $S' \subseteq S$ be the set of $(i,j,k) \in S$ such that $\alpha_{ijk}$ is not an integer multiple of $\frac{1}{n}$.
For each $(i,j,k) \in S \setminus S'$, we pick $\alpha'_{ijk} = \alpha_{ijk}$. Initially, for each $(i,j,k) \in S'$, let $\alpha'_{ijk}$ be $\alpha_{ijk}$ rounded \emph{down} to the next integer multiple of $\frac{1}{n}$, i.e., pick $\alpha'_{ijk} = \frac{\lfloor n \cdot \alpha_{ijk} \rfloor}{n}$. Since $\alpha \in D$ we have that $\sum_{(i,j,k) \in S} \alpha_{ijk} = 1$. Let $t = \sum_{(i,j,k) \in S} \alpha'_{ijk}$. $t$ is an integer multiple of $\frac{1}{n}$. Moreover, we have that
$$1 - t = 1 - \sum_{(i,j,k) \in S} \alpha'_{ijk} = \sum_{(i,j,k) \in S} \alpha_{ijk} - \alpha'_{ijk} =  \sum_{(i,j,k) \in S'} \alpha_{ijk} - \alpha'_{ijk} \leq \sum_{(i,j,k) \in S'} \frac{1}{n} = \frac{|S'|}{n}.$$
There is hence a nonnegative integer $K \leq |S'|$ such that $t = 1 - \frac{K}{n}$. Pick any $K$ elements $(i,j,k)$ of $S'$ and add $\frac{1}{n}$ to $\alpha'_{ijk}$. We now have $\sum_{(i,j,k) \in S} \alpha'_{ijk} = 1$, and so $\alpha' \in D$.

We always have that $\alpha'_{ijk}$ is an integer multiple of $\frac{1}{n}$, since we originally picked integer multiples of $\frac{1}{n}$, and then possibly added $\frac{1}{n}$ to them. Finally, we always have $|\alpha_{ijk} - \alpha'_{ijk}| < \frac{1}{n}$ since we initially rounded each $\alpha_{ijk}$ for $(i,j,k) \in S'$ down to the next integer multiple of $\frac{1}{n}$, then possibly added $\frac{1}{n}$ to it.
\end{proof}

\begin{lemma}
Let $\alpha,\alpha' \in D$ be as above. Then, $$\frac{1}{2^{o(n)}} \leq \frac{\alpha_N}{\alpha'_N} \leq 2^{o(n)}.$$
\end{lemma}

\begin{proof}
Recall that $$\alpha_N := \prod_{(i,j,k) \in S} \alpha_{ijk}^{-\alpha_{ijk}}.$$ Thus, $$\frac{\alpha_N}{\alpha'_N} = \prod_{(i,j,k) \in S} \frac{\alpha_{ijk}^{-\alpha_{ijk}}}{{\alpha'}_{ijk}^{-{\alpha'}_{ijk}}}.$$
Since $S$ is a constant-sized set, it is sufficient to show that for any fixed $(i,j,k) \in S$ we have $$\frac{1}{2^{o(n)}} \leq \frac{\alpha_{ijk}^{-\alpha_{ijk}}}{{\alpha'}_{ijk}^{-\alpha'_{ijk}}}\leq 2^{o(n)}.$$
First, if $\alpha_{ijk}$ is an integer multiple of $1/n$, including if $\alpha_{ijk}=0$, then we have $\alpha'_{ijk} = \alpha_{ijk}$, and so $\frac{\alpha_{ijk}^{-\alpha_{ijk}}}{{\alpha'}_{ijk}^{-{\alpha'}_{ijk}}} = 1$.
Otherwise, let $\delta = \alpha'_{ijk} - \alpha_{ijk}$, so we have $0 < |\delta| < \frac{1}{n}$. Consider first when $\delta>0$. We can write 
$$\log\left( \frac{\alpha_{ijk}^{-\alpha_{ijk}}}{{\alpha'}_{ijk}^{-\alpha'_{ijk}}} \right) = (\alpha_{ijk} + \delta) \log(\alpha_{ijk}+\delta) - \alpha_{ijk} \log \alpha_{ijk} = \alpha_{ijk} \log\left(1 + \frac{\delta}{\alpha_{ijk}}\right) + \delta \log(\alpha_{ijk} + \delta).$$
Using the fact that $\log(1+x) = x - O(x^2)$ as $x \to 0^+$, and that $\alpha_{ijk}$ is a positive constant, we can bound $$\alpha_{ijk} \log\left(1 + \frac{\delta}{\alpha_{ijk}}\right) \leq \alpha_{ijk} \left( \frac{\delta}{\alpha_{ijk}} + O(\delta^2) \right) \leq O(\delta) = O\left(\frac{1}{n}\right).$$
For large enough $n$, we also have
$$\delta \log(\alpha_{ijk} + \delta) \leq \delta \log(2\alpha_{ijk}) \leq O\left(\frac{1}{n}\right).$$
It follows that $$\frac{\alpha_{ijk}^{-\alpha_{ijk}}}{{\alpha'}_{ijk}^{-\alpha'_{ijk}}} \leq 2^{O(1/n)}.$$ We have $\frac{\alpha_{ijk}^{-\alpha_{ijk}}}{{\alpha'}_{ijk}^{-{\alpha'}_{ijk}}} \geq 1$ since $\delta>0$, which completes the proof for $\delta>0$. The proof for $\delta<0$ is nearly identical. 
\end{proof}

Very similar proofs show that, for $\alpha,\alpha' \in D$ as above, the ratios $\frac{\alpha_{V{\tau}}}{\alpha'_{V_\tau}}$, $\frac{\alpha_B}{\alpha'_B}$, and $\frac{\max_{\beta\in D_\alpha} \beta_N}{\max_{\beta\in D_{\alpha'}} \beta_N}$, are all bounded between $2^{-o(n)}$ and $2^{o(n)}$ as well.

\section{Lemma for Section~\ref{subsec:step3}}\label{app:powermean}

We now prove Lemma~\ref{lem:apppowermean}, which was used in the proof in Section~\ref{subsec:step3} above. To apply it there, $i = 1, \ldots, c$ are the different choices of randomness (i.e., the choices of $w_0, \ldots, w_{2n}$) defining the hash functions.

\begin{lemma} \label{lem:apppowermean}
Suppose $c$ is a positive integer, and $A,B,a_1, \ldots, a_c, b_1, \ldots, b_c$ are positive real numbers such that $\sum_{i=1}^c a_i = c \cdot A$ and $\sum_{i=1}^c b_i = c \cdot B$. Then, there exists an $i \in [c]$ such that $$\frac{a_i^{3/2}}{b_i^{1/2}} \geq \frac{A^{3/2}}{B^{1/2}}.$$
\end{lemma}

\begin{proof}
It suffices to prove that \begin{equation}\label{eqtoprove}\sum_{i=1}^c \frac{a_i^3}{b_i} \geq c \cdot \frac{A^3}{B}.\end{equation} Indeed, if (\ref{eqtoprove}) is true, then by the pigeonhole principle, there must exist an $i$ which achieves at least the average value $$\frac{a_i^3}{b_i} \geq \frac{A^3}{B},$$ and hence by taking square roots, $$\frac{a_i^{3/2}}{b_i^{1/2}} \geq \frac{A^{3/2}}{B^{1/2}},$$ as desired.

To prove (\ref{eqtoprove}), we first apply the power mean inequality, then combine this with the Cauchy-Schwarz inequality. The power mean inequality says that
$$\left( \frac{1}{c} \sum_{i=1}^c a_i^{3/2} \right)^{2/3} \geq \frac{1}{c} \sum_{i=1}^c a_i = A,$$
and so by cubing both sides,
\begin{equation}\label{eqpowermean} \left( \frac{1}{c} \sum_{i=1}^c a_i^{3/2} \right)^{2} \geq  A^3. \end{equation}
Next we apply the Cauchy-Schwarz inequality, which says
$$\left( \sum_{i=1}^c \frac{a_i^3}{b_i} \right) \left( \sum_{i=1}^c b_i \right) \geq \left( \sum_{i=1}^c a_i^{3/2} \right)^2.$$
Combining this with (\ref{eqpowermean}) gives
$$\left( \sum_{i=1}^c \frac{a_i^3}{b_i} \right)  \geq \frac{\left( \sum_{i=1}^c a_i^{3/2} \right)^2}{ \sum_{i=1}^c b_i } = \frac{c^2 \cdot \left( \frac{1}{c}\sum_{i=1}^c  a_i^{3/2} \right)^2}{c \cdot B} \geq \frac{c \cdot A^3}{B},$$ as desired.
\end{proof}

\end{document}